\documentclass[aps,prd,reprint,onecolumn,nofootinbib]{revtex4-2}
\usepackage[utf8]{inputenc}
\usepackage{amsmath}
\usepackage{dsfont}
\usepackage{mathrsfs}
\usepackage{tikz}
\usetikzlibrary{matrix}
\usepackage[absolute,overlay]{textpos}
\usepackage{cases}
\usepackage{dutchcal}
\usepackage{cancel}
\usepackage[normalem]{ulem}
\usepackage{amsmath,amssymb}
\makeatletter
\newsavebox\myboxA
\newsavebox\myboxB
\newlength\mylenA

\newcommand*\xoverline[2][0.75]{%
    \sbox{\myboxA}{$\m@th#2$}%
    \setbox\myboxB\null% Phantom box
    \ht\myboxB=\ht\myboxA%
    \dp\myboxB=\dp\myboxA%
    \wd\myboxB=#1\wd\myboxA% Scale phantom
    \sbox\myboxB{$\m@th\overline{\copy\myboxB}$}%  Overlined phantom
    \setlength\mylenA{\the\wd\myboxA}%   calc width diff
    \addtolength\mylenA{-\the\wd\myboxB}%
    \ifdim\wd\myboxB<\wd\myboxA%
       \rlap{\hskip 0.5\mylenA\usebox\myboxB}{\usebox\myboxA}%
    \else
        \hskip -0.5\mylenA\rlap{\usebox\myboxA}{\hskip 0.5\mylenA\usebox\myboxB}%
    \fi}
\makeatother
\DeclareMathAlphabet{\mathdutchcal}{U}{dutchcal}{m}{n}
\SetMathAlphabet{\mathdutchcal}{bold}{U}{dutchcal}{b}{n}
\DeclareMathAlphabet{\mathdutchbcal}{U}{dutchcal}{b}{n}

\usepackage{orcidlink}
\usepackage{enumitem}
\usepackage{etoolbox}

\usepackage{chngcntr}

\usepackage{slashed}
\usepackage{pgf} 
\usepackage{stmaryrd}
\usepackage{hyperref}
\hypersetup{
	colorlinks,
	citecolor=violet,
	linkcolor=red,
	urlcolor=blue}
\usepackage{mathtools}

\usepackage{etex}

\usepackage{grffile}

\usepackage{accents}
\usepackage{bm}
\usepackage{enumitem}
\usepackage{amssymb}
\usepackage{extarrows}
\usepackage{calrsfs}
\usepackage{amsthm}

\usepackage{geometry}
\usepackage[T1]{fontenc}
\usepackage{calligra}
\DeclareMathAlphabet{\mathpzc}{OT1}{pzc}{m}{it}
\usepackage{array,multirow,makecell}
\setcellgapes{1pt}
\makegapedcells
\newcolumntype{R}[1]{>{\raggedleft\arraybackslash }b{#1}}
\newcolumntype{L}[1]{>{\raggedright\arraybackslash }b{#1}}
\newcolumntype{C}[1]{>{\centering\arraybackslash }b{#1}}

\newtheorem{theorem}{Theorem}[section]
\newtheorem{corollary}{Corollary}[theorem]
\newtheorem{prop}{Proposition}[section]

\theoremstyle{definition}
\newtheorem{definition}{Definition}[section]

\theoremstyle{remark}
\newtheorem*{remark}{Remark}

\newcommand\numberthis{\addtocounter{equation}{1}\tag{\theequation}}

\newcommand{\X}{\mathfrak{X}}

\newcommand{\cm}{\mathcal{M}}
\newcommand{\p}{\mathfrak{p}}
\newcommand{\bp}{{\bar{\mathfrak{p}}}}
\newcommand{\h}{\mathfrak{h}}
\newcommand{\g}{\mathfrak{g}}
\newcommand{\m}{\mathfrak{m}}

\DeclareMathOperator{\Tr}{Tr}
\DeclareMathOperator{\tr}{tr}
\DeclareMathOperator{\Pf}{Pf}
\DeclareMathOperator{\Ad}{Ad}
\DeclareMathOperator{\ad}{ad}

\newcommand{\GamTtM}{\Gamma(T\widetilde{\mathcal{M}})}

\newcommand{\tg}{\tilde{g}}

\newcommand{\tS}{\tilde{S}}

\newcommand{\varpih}{\varpi_\mathfrak{h}}

\newcommand{\varpiz}{\varpi_{0}}

\newcommand{\varpip}{\varpi_\p}

\newcommand{\varpim}{\varpi_\m}

\newcommand{\varpihmod}[2][ ]{\varpi^{ #1}_{\h \:\,#2}}

\newcommand{\Amod}[2][ ]{A^{#1}_{~\,\,#2}}

\newcommand{\alphamod}[2][ ]{\alpha^{ ~\,#1}_{~~\,\,#2}} 

\newcommand{\betamod}[2][ ]{\beta^{\,#1}_{~\,\,#2}} 

\newcommand{\varpipmod}[2][ ]{\varpi^{ ~\,#1}_{\p\,\,#2}}

\newcommand{\betainvmod}[2][ ]{\beta^{-1 \, #1}_{~~\,\,\; \; \;#2}}

\newcommand{\Rmod}[2][ ]{R^{#2}_{\;\; \; #1}}

\newcommand{\om}{\omega}

\newcommand{\Om}{\Omega}

\newcommand{\bOm}{\bar{\Omega}}

\newcommand{\bOmh}{\bar{\Omega}_{\h}}

\newcommand{\bOmm}{\bar{\Omega}_{\m}}

\newcommand{\bOmp}{\bar{\Omega}_{\p}}

\makeatletter
\def\lefteqno{\tagsleft@true}\def\righteqno{\tagsleft@false}
\makeatother 

\makeatletter
\let\save@mathaccent\mathaccent
\newcommand*\if@single[3]{%
	\setbox0\hbox{${\mathaccent"0362{#1}}^H$}%
	\setbox2\hbox{${\mathaccent"0362{\kern0pt#1}}^H$}%
	\ifdim\ht0=\ht2 #3\else #2\fi
}
\newcommand*\rel@kern[1]{\kern#1\dimexpr\macc@kerna}
\newcommand*\widebar[1]{\@ifnextchar^{{\wide@bar{#1}{0}}}{\wide@bar{#1}{1}}}
\newcommand*\wide@bar[2]{\if@single{#1}{\wide@bar@{#1}{#2}{1}}{\wide@bar@{#1}{#2}{2}}}
\newcommand*\wide@bar@[3]{%
	\begingroup
	\def\mathaccent##1##2{%
		\let\mathaccent\save@mathaccent
		\if#32 \let\macc@nucleus\first@char \fi
		\setbox\z@\hbox{$\macc@style{\macc@nucleus}_{}$}%
		\setbox\tw@\hbox{$\macc@style{\macc@nucleus}{}_{}$}%
		\dimen@\wd\tw@
		\advance\dimen@-\wd\z@
		\divide\dimen@ 3
		\@tempdima\wd\tw@
		\advance\@tempdima-\scriptspace
		\divide\@tempdima 10
		\advance\dimen@-\@tempdima
		\ifdim\dimen@>\z@ \dimen@0pt\fi
		\rel@kern{0.6}\kern-\dimen@
		\if#31
		\overline{\rel@kern{-0.6}\kern\dimen@\macc@nucleus\rel@kern{0.4}\kern\dimen@}%
		\advance\dimen@0.4\dimexpr\macc@kerna
		\let\final@kern#2%
		\ifdim\dimen@<\z@ \let\final@kern1\fi
		\if\final@kern1 \kern-\dimen@\fi
		\else
		\overline{\rel@kern{-0.6}\kern\dimen@#1}%
		\fi
	}%
	\macc@depth\@ne
	\let\math@bgroup\@empty \let\math@egroup\macc@set@skewchar
	\mathsurround\z@ \frozen@everymath{\mathgroup\macc@group\relax}%
	\macc@set@skewchar\relax
	\let\mathaccentV\macc@nested@a
	\if#31
	\macc@nested@a\relax111{#1}%
	\else
	\def\gobble@till@marker##1\endmarker{}%
	\futurelet\first@char\gobble@till@marker#1\endmarker
	\ifcat\noexpand\first@char A\else
	\def\first@char{}%
	\fi
	\macc@nested@a\relax111{\first@char}%
	\fi
	\endgroup
}
\makeatother

\setlist[enumerate]{label=\thesection.\arabic{*},resume}

\preto\section{%
	\restartlist{enumerate}%
}

\begin{document}
	\author{
		Jean Thibaut\,\orcidlink{0009-0004-7396-1395}\footnote{jthibaut@cpt.univ-mrs.fr} 
		$\ $and Serge Lazzarini\footnote{serge.Lazzarini@cpt.univ-mrs.fr} \\
		{\normalsize Centre de Physique Théorique}\\
		{\normalsize Aix Marseille Univ, Université de Toulon, CNRS, CPT, Marseille, France.}
	}
	
	\date{\footnotesize \it \today}
	
	\title{Gravity as a deformed
		topological gauge theory}
%	\maketitle

	\begin{abstract}
		We describe gauge theories which allow to retrieve a large class of gravitational theories, including, MacDowell-Mansouri gravity and its topological extension to Loop Quantum Gravity via the Pontrjagin characteristic class involving the Nieh-Yan term. 
     Considering symmetric spaces parametrized by mutations allows to naturally obtain a bare cosmological constant which in particular cases gives rise to a positive effective cosmological constant while having an AdS spacetime.

    Two examples are studied, Lorentzian geometry (including dS and AdS spacetimes) and Lorentz$\times$Weyl geometry.
    In the latter case, we prove that adding a term dependent on torsion and dilations makes the equations of motion exhibit a secondary source for curvature in addition to the usual energy-momentum tensor. This additional source is expressed in terms of the spin density of matter, torsion and their variations.

    Finally, we show that the gauge + matter actions constructed from invariant polynomials are asymptotically topological if one assumes a vanishing bare cosmological constant together with gauge and matter fields having compact support.
	\end{abstract}
    \maketitle
	%\keywords{Topological gravity, Homogeneous spaces, Cartan geometry, Characteristic classes}
    
%	\newpage
	
	\section{Introduction}
	
	There exist several approaches to describe gravity as a gauge theory.
	Two of them are MacDowell-Mansouri (MM) gravity and BF theory, see \emph{e.g.}~\cite{wise_macdowell-mansouri_2010}.
	They can be mathematically described in terms of Cartan geometries on a manifold $\mathcal{M}$ modeled on a homogeneous space $G/H$ where $G,H\subset G$ and $\g,\h$ are respectively Lie groups and their corresponding Lie algebras.
	When the geometry is reductive, namely, $\g=\h\oplus\p$, a $\Ad (H)$-invariant decomposition, the associated $\g$-valued Cartan connection $\varpi$ naturally splits into two parts $\varpi = \varpi_\h \oplus\varpi_\p$, one with values in $\h$ and the other one in $\p=\g / \h$.
	In these formalisms the $\h$-valued part $\varpih$ corresponds to the spin connection $A$ while the $\p$-valued part is arranged to be $\varpip=
        \sqrt{\dfrac{\Lambda_0}{3\epsilon}} \begin{pmatrix}
			0 &  \beta\\
			- \epsilon \bar{\beta} & 0
		\end{pmatrix} $ and is proportional to the vielbein 1-form $\beta$ (tetrad if $dim(\mathcal{M})=4$) where $\bar{\beta} = \eta(\beta)$ with $\eta$ the Minkowski metric.
	
	\smallskip
	For a resulting Lorentzian\footnote{Riemaniann or Galilean spacetimes can also be described by this formalism. For more details see~\cite[p.10]{wise_macdowell-mansouri_2010}. The kind of topological gauge theories developed in this article are also compatible with these choices of Lie groups.} space-time the structure group is given by $H=SO(3,1)$ while the Lie group $G$ is chosen according to the value and sign of the (bare) cosmological constant $\Lambda_0$ as in \cite[p.10]{wise_macdowell-mansouri_2010}:
		\begin{align}
		\left\{
		\begin{array}{lcll}
			G = SO(4,1) & \text{ and }& \epsilon = +1, & \text{ for } \Lambda_0 > 0 \\
			G = SO(3,2) & \text{ and } & \epsilon = -1, & \text{ for } \Lambda_0 < 0 .
		\end{array}
        \right.
	\end{align}
	
	\medskip
	Let us recall the
	\begin{definition}
		Let $\mathcal{M}$ be a manifold, $G$ a Lie group and $H$ a closed subgroup of $G$ such that $dim(G/H)=dim(\mathcal{M})=m$ and
		$\mathcal{P}$ a principal $H$-bundle over $\mathcal{M}$. Let $\mathfrak{X}\in \h$, we denote $\mathfrak{X}^v$ the associated vertical vector field on $\mathcal{P}$.
		A Cartan connection  on $\mathcal{P}$ is equivalent to a $1$-form $\varpi$ on $\mathcal{P}$ with values in $\g$, the Lie algebra of $G$, satisfying the conditions \cite[p128]{kobayashi_transformation_1995}:
		\begin{align*}
			&(a) ~ \varpi(\X^v) = \X,~ \text{ for any } \X \in \h. \\
			&(b) ~ (R_h)^* \varpi = \Ad(h^{-1}) \varpi ~, \forall h \in H ,\text{ where } R \text{ is the right action of } H \text{ on } \mathcal{P} \\
			&\quad\ \, \text{and } \Ad \text{ the adjoint representation of } H\subset G \text{ on the Lie algebra } \g. \notag \\
			&(c) ~ \varpi(X) \neq 0 \text{ for every nonzero vector field } X \in \text{Vect}(\mathcal{P}).
		\end{align*}
	\end{definition}
In particular, condition (c) implies that tangent spaces to $\mathcal{M}$ are isomorphic to the quotient $\g/\h$ of dimension $m$.
    
	\medskip
	Let $\bOm = d\varpi + \varpi \wedge \varpi = \bOmh \oplus \bOmp $ be the curvature of the Cartan connection $\varpi= \varpi_\h \oplus\varpi_\p$ and $\star$ an internal Hodge as used in \cite{montesinos_self-dual_2001} such that $\star \kappa^{ab} = \tfrac{1}{2} \varepsilon^{ab}_{~~kl} \kappa^{kl} $.
	
	\medskip
	For $m=4$, the MacDowell-Mansouri action \cite{wise_macdowell-mansouri_2010} is given by:
	\begin{align} \label{eq:MM}
		S_{MM}[\varpi] &=  - \dfrac{3}{2G\Lambda_0}\int_{\mathcal{M}}  \tr(\bOmh \wedge \star \bOmh) = - \dfrac{3}{4G\Lambda_0}\int_{\mathcal{M}} \varepsilon_{abcd} \,\bOmh^{ab} \wedge \bOmh^{cd}
	\end{align}
    where $\bOmh = R - \dfrac{\Lambda_0}{3} \beta \wedge \bar{\beta}$.
	This action is $\Ad (H)$-invariant, {\em i.e.} Lorentz invariant.
	
	\medskip
	On the other hand, the Euler characteristic class $E(\mathcal{M})$ of a manifold $\mathcal{M}$ is null if $\dim(\mathcal{M})=m$ is odd, otherwise it can be linked to the $\tfrac{m}{2}$-th Pontrjagin class $P_{m/2}$ 
	(more details can be found in \cite{nakahara_geometry_2018}).
	\begin{definition}
		In the case of an orientable differentiable manifold $\mathcal{M}$ of \emph{even} dimension $m$, the Euler characteristic class can always \footnote{Since as indicated in \cite{nakahara_geometry_2018} it is always possible to reduce the structure group of $T\mathcal{M}$ to $SO(m)$ via the Gram-Schmidt process.} be expressed in terms of curvature as:
		\[
		E(\mathcal{M}) = \Pf\left(\dfrac{\Om_{\mathfrak{so}}}{2\pi}\right) 
		\]
		where $\Pf(\Om_{\mathfrak{so}})$ is the Pfaffian of the $\mathfrak{so}(m)$-valued curvature 2-form $\Om_{\mathfrak{so}}$ originating from the $\mathfrak{so}(m)$-valued connection $1$-form $\om_{\mathfrak{so}}$. This is an $SO(m)$-invariant polynomial. 
	\end{definition}
	For a $4$-dimensional manifold one has: 
	\[
	E(\mathcal{M}) = \tfrac{1}{2(4\pi)^2}\, \varepsilon_{abcd}\, \Om_{\mathfrak{so}}^{ab} \wedge \Om_{\mathfrak{so}}^{cd}
	\]
	More generally, let us give the
	\begin{definition}
		The characteristic number $\mathdutchcal{C}$ associated to a given characteristic class $C(\mathcal{M})$ of $\mathcal{M}$ corresponds to the integral over $\mathcal{M}$ of its characteristic polynomial :
		\[
		\mathdutchcal{C} = \int_{\mathcal{M}} C(\mathcal{M}).
		\]
	\end{definition}
    In particular, for $m=4$, the Euler number is thus:
	\begin{align}
		\label{Euler number}
		\mathdutchcal{E} = \dfrac{1}{2(4\pi)^2} \int_{\mathcal{M}} 
		\varepsilon_{abcd}\, \Om_{\mathfrak{so}}^{ab} \wedge \Om_{\mathfrak{so}}^{cd}.
	\end{align}
	One can recognize the MacDowell-Mansouri action \eqref{eq:MM} is proportional to the Pfaffian of the $\h=\mathfrak{so}(3,1)$-valued part $\bOmh$ of the curvature, which is $SO(3,1)$-invariant. This provides a very clear mathematical interpretation to the MacDowell-Mansouri action as the deformation of a Schwarz-type\footnote{A more detailed account of topological gauge theories of Schwarz-type can be found in \cite{kaul_schwarz_2005,birmingham_topological_1991,bartlett2005categoricalaspectstopologicalquantum}. The Euler number also yields a topological field theory according to \cite{atiyah_topological_1990}. See also \cite{thuillier_remarks_1998}.} topological gauge theory:
	\begin{align}
    \label{eq:MM-Pf}
		S_{MM}[\varpi] = - \dfrac{24\pi^2}{ G \Lambda_0}\, \int_\mathcal{M} \Pf (\dfrac{\bOmh}{2\pi})
	\end{align}
	where $\bOmh = R - \dfrac{\Lambda_0}{3} \beta \wedge \bar{\beta}$, and $- \dfrac{\Lambda_0}{3} \beta \wedge \bar{\beta}$ acts as the deformation of the $\bOm_\mathfrak{so} = R$ curvature.
 
	\medskip
	Considering Chern-Simons theory as an example of topological field theory of the Schwarz type, another characteristic class which will be of interest in this work is the Pontrjagin characteristic class. 		
	As described in \cite{nakahara_geometry_2018}, one has:
	\begin{definition}
		Let $L$ be a $k$-dimensional real vector bundle over an $m$-dimensional manifold $\mathcal{M}$. $\Om$ is the curvature associated to a connection $\omega$ on a principal $G$-bundle over $\mathcal{M}$.
		Introducing orthonormal frames allows to reduce the structure group to $O(k)$ and one has $\mathfrak{o}(k) = \mathfrak{so}(k)$ at the level of Lie algebras.
		The $\mathfrak{so}(k)$-valued curvature is written $\Om_{\mathfrak{so}}$.
		The corresponding total Pontrjagin class is then defined as: $P(\Om_{\mathfrak{so}}) = \det\Bigl(\mathds{1} + \dfrac{\Om_{\mathfrak{so}}}{2\pi} \Big)$.
	\end{definition}
	
	For a 4-dimensional manifold $\mathcal{M}$, the Pontrjagin number turns out to be: 
	\begin{align}
		\label{Pontrjagin number}
		\mathdutchcal{P} = \int_{\mathcal{M}} P(\Om_{\mathfrak{so}}) = \int_{\mathcal{M}} P_1(\Om_{\mathfrak{so}}) = \int_{\mathcal{M}} -\frac{1}{8\pi^2} \Tr( \Om_{\mathfrak{so}} \wedge \Om_{\mathfrak{so}})\ ,
	\end{align}
	where $P_1$ is the $1^{\text{st}}$ Pontrjagin class.
	
	As stated in \cite[Chern-Weil theorem]{nakahara_geometry_2018}, a nice feature of characteristic classes is that up to an exact term, they do not depend on the choice of connection. This fact allows us to build actions that are invariant under the choice of connection.

   Generalizing the observation made for \eqref{eq:MM-Pf}, the gauge invariant actions studied in this work are given by the linear combination:
    \begin{align} \label{eq:action_start}
       S_G [\varpi] = \int_\mathcal{M} \big(r P(\bOm) + e \Pf  (\dfrac{F}{2\pi}) + y \det ( \mathds{1} + \dfrac{\bOmh}{2\pi})\big)
    \end{align}
where the first term is the Pontryagin number of the manifold related to the $\g$-valued Cartan curvature $\bOm$ while the two other terms are only invariant polynomials corresponding to the Pfaffian $\Pf (\dfrac{F}{2\pi})$ and the determinant $\det ( \mathds{1} + \dfrac{\bOmh}{2\pi})$. They are respectively evaluated on $F$ (the $\mathfrak{so}(m-1,1)$-valued part of the curvature $\bOmh$) and $\bOmh$ ($\h$-valued) parts of the Cartan curvature.    
       
If $\g/\h = \m$ obeys the commutation relation $[\m,\m] \subset \h$ (symmetric Lie algebra), then $G/H$ is called a reductive homogeneous space \cite{donnelly1977chern}. In fact, this is the case for the construction given in \cite{wise_macdowell-mansouri_2010} for the MacDowell-Mansouri action. We shall extend this idea, and we shall see that Cartan geometries modeled on reductive homogeneous spaces are intrinsically linked to the notion of cosmological constant in specific cases. In the examples given in this paper, reductive homogeneous spaces will be implemented at the infinitesimal level through symmetric Lie algebras. Using the concept of mutation, two symmetric Lie algebras where $H$ is either the Lorentz or Lorentz$\times$Weyl group will be given as resulting from a mutation procedure on which a Cartan geometry will be constructed along the line of \cite{sharpe2000differential}.

    \bigskip
    The paper is organized as follows.
	Section~\ref{section SO 5 ETC} will introduce and describe the gauge theories associated to the Lie groups $G = SO(4,1)$ (respectively $SO(3,2)$) for $\Lambda_0 > 0$ (respectively $\Lambda_0 < 0$) considered as "mutations" (according to \cite{sharpe2000differential}) of the group $G=ISO(3,1)$. All these three groups contain $H=SO(3,1)$ which will be the structure group of the principal bundle $\mathcal{P}$. Then, the gauge invariant action \eqref{eq:action_start} will be studied for these Lie groups $G$ and $H$. 
   Subsenquently, a very particular linear combination yields the Holst + Euler and Pontrjagin of the curvature $R$ of the spin connection + Nieh-Yan + bare cosmological constant $\Lambda_0$ terms.
    Additionally, in this construction, the coupling constants of these different terms are inherently linked together.        
    Results of Section~\ref{section SO 5 ETC} are adaptable to other spacetimes whether Riemannian or Galilean or for other Lie groups. 
 
	Section~\ref{section mobius} will be devoted to the construction of a similar action for the case where $G=SO(4,2)/\{\pm I\}$ (the Möbius group) and $H = CO(3,1) = SO(3,1)\times(\mathbb{R}_+\setminus \{0\})$ the Lorentz$\times$Weyl group of signature $(3,1)$, so that the infinitesimal Klein geometry $(\g,\h)$ in the sense of \cite{sharpe2000differential} turns out to be a symmetric Lie algebra. Then, mutating the $(\g,\h)$ model allows to mimic what we did in the previous Lorentz case. Reducing the quotient $\g/\h$ to be of dimension $4$ by constraining  the "pair of frames" $(\alpha,\beta)$ one can retrieve an action comprising all terms described in the previous example + a kinetic term for a scalar field (dilations).

  In both examples, we study the equations of motion associated to the total action (gauge + matter). It is especially shown that, in the Möbius case, if one adds an interaction 
  between dilations and torsion $T$, the equations of motion consist of Einstein’s equations modified by the Holst term with an additional source term for curvature depending on specific variations of spin density of matter and on torsion.
		
   Lastly, in Section \ref{Topological total action}, by studying invariant polynomials for reductive homogeneous spaces, it will be shown that (upon restricting to connections with compact supports on $\mathcal{M}$) the mutated geometries can be considered as deformations such that the gauge invariant actions constructed in the previous sections are asymptotically invariant under the choice of this type of connections, and thus lead to topological theories.
 
	\smallskip
	In the following, latin and greek letters will respectively be used for algebraic and space-time indices.
	
	\section{Invariant polynomials for a Cartan connection with Lorentzian spacetimes}
	
	\label{section SO 5 ETC}
	
	\subsection{Lorentzian (dS and AdS) spacetimes as mutations of an \\ $ISO(3,1)/SO(3,1)$ Cartan geometry}

	Let $G$, $H\subset G$ be Lie groups with corresponding Lie algebras $\g =\h \oplus \p$
	such that $G/H$ is reductive \cite{sharpe2000differential}, that is that the splitting $\g =\h \oplus \p$ is $\Ad(H)$-invariant.
	
	One such reductive model geometry is $G/H = ISO(3,1)/SO(3,1)$, with algebra $\g=\overbrace{\mathfrak{so}(3,1)}^\h \oplus \overbrace{\mathfrak{\mathbb{R}}^{3,1}}^\p $ where
	$G=ISO(3,1)= SO(3,1) \ltimes \mathbb{R}^{3,1}$ and $H=SO(3,1)$ is defined by the relation 
	$\Lambda^{\text{T}} \eta \Lambda =~\eta, \forall \Lambda \in SO(3,1)$ with $\eta$ the Minkowski metric:
	\begin{align*}
		\eta = \begin{pmatrix}
			-1 & 0 \\
			0 & +I_3
		\end{pmatrix}\ .
	\end{align*}
	A group element $g \in G=ISO(3,1)$ is parametrized as
	${\displaystyle 
		g = \begin{pmatrix}
			\Lambda &  a\\
			0 & 1
		\end{pmatrix}
	}$
	where $a \in \mathbb{R}^{3,1}$.
		Linearization leads to:
	\begin{align*}
		\varphi &=
		\begin{pmatrix}
			\varphi_\h & \varphi_\p \\
			0 & 0
		\end{pmatrix}
		\in \mathfrak{iso}(3,1)
		\text{ where } \varphi_\h \in \mathfrak{so}(3,1)~,~ \varphi_\h^{\text{T}} \eta + \eta \varphi_\h =0, ~ \text{and}~ \varphi_\p \in \mathbb{R}^{3,1}\\[2mm]
		&= \begin{pmatrix}
			\frac{1}{2}\varphi^{ab}_\h J_{ab} & \varphi^a_\p e_a \\
			0 & 0
		\end{pmatrix} \in \mathfrak{iso}(3,1)
	\end{align*}
    with $\{J_{ab}\}_{a,b=0,...,3}$ the generators of $\h =  \mathfrak{so}(3,1)$ and $\{e_a\}_{a=0,...,3}$ a basis of $\p = \mathbb{R}^{3,1}$.

    \medskip
	In \cite[see p.218]{sharpe2000differential}, Sharpe introduces the notion of mutation of a Cartan geometry $G/H$ according to the	
	\begin{definition} \label{defmutation}
		Let $(\g,\h)$ and $(\g',\h)$ be two model geometries.
		A mutation map corresponds to an $\Ad (H)$ module isomorphism $\mu : \g \rightarrow \g'$ ({\sl i.e.} $\mu(\Ad_h(u)) = \Ad_h(\mu(u))\, \forall u\in\g$) satisfying:
		\begin{flalign*}
			(i) &\ \mu_{|\h} = \mathrm{id}_\h \\
			(ii) &\ \bigl[ \mu(u) , \mu(v) \bigr] = \mu \bigl([ u,v ] \bigr) \text{ mod } \h, \forall u,v \in \g.
		\end{flalign*}
		The model geometry $(\g',\h)$ then corresponds to the mutant of the model geometry $(\g,\h)$ with the same group $H$.
	\end{definition}
	
	In the case where the model geometry is $(\g,\h)= \bigr(\mathfrak{iso}(3,1), \mathfrak{so}(3,1)\bigr)$, a mutation that will be particularly relevant for us is given by the mutation map\footnote{This mutation map generalizes the examples given in \cite[Example 6.2, p.218]{sharpe2000differential}.} $\mu : \g \rightarrow \g'$:
 \begin{align}
		\label{mutation}
		\varphi_\g =
		\begin{pmatrix}
			\varphi_\h	&  \varphi_\p \\
			0 & 0
		\end{pmatrix}
		\mapsto
		\varphi_{\g'} = \mu(\varphi_\g)=
		\begin{pmatrix}
			\varphi_\h	& k_1 \varphi_\p \\
			k_2 \bar{\varphi}_{\p} & 0
		\end{pmatrix} \in \g'
	\end{align}
	with $\bar{\varphi}_{\p} =  \varphi_\p^{\text{T}}\eta$ and where $k_1,k_2$ may a priori be considered as $0$-form scalar fields.
 It is worthwhile to notice that this mutation transforms the geometry modeled on $G/H$ given earlier to the Lorentzian Cartan geometry modeled on $G'/H$ with elements $g' \in G'$ obeying the group relation $g'^TNg' = N$, where $N$ is the following metric: $$N=
    \begin{pmatrix}
			\eta & 0 \\
			0 & -k_1/k_2
	\end{pmatrix}\ $$
    
 \smallskip\noindent
One can observe that $ (\g',\h) $ reduces to dS and AdS geometries respectively for the two particular cases, see {\sl e.g.} \cite{wise_macdowell-mansouri_2010}:	\begin{align}
		\left\{
		\begin{array}{lcll}
			\dfrac{k_1}{k_2}= -1\Leftrightarrow  G = SO(4,1) & \text{ and }& \epsilon = +1, & \text{ for } \Lambda_0 > 0, \\
			& & &\\
			\dfrac{k_1}{k_2}=+1 \Leftrightarrow G = SO(3,2) & \text{ and } & \epsilon = -1, & \text{ for } \Lambda_0 < 0.
		\end{array}
		\right.
	\end{align}
	Throughout the paper $k_1$ and $k_2$ will be considered as real numbers.
    It is however of physical interest to consider the mutation parameters as scalar fields, thus dealing with point-wisely dependent mutations. This is studied in \cite{Thibaut:2025lwx}. Other papers considering dynamical versions of physical constants in the context of gravity include for instance \cite{alexander_zero-parameter_2019,sengupta2025cosmological}.
	
	\medskip
	On the other hand, the canonical decomposition of a symmetric Lie algebra $\g = \h \oplus \m$ (see \cite[Prop.2.1, Chapter XI]{kobayashi_foundations_1996}
 for more examples) corresponds to the commutators:
	\begin{align*}
		[\h,\h] \subset \h, \qquad [\h,\m] \subset \m, \qquad [\m,\m] \subset \h.
	\end{align*}
	Therefore, one can see that the mutated Lie algebra $\g'= \h  \oplus \m$ of \eqref{mutation} corresponds to a symmetric Lie algebra for $k \neq 0$, as a particular case of a reductive geometry with $\m = \g'/\h \simeq \p = \g/\h$.
	\label{generators}
	Moreover, the generators of $\h = \mathfrak{so}(3,1)$, and elements of the basis of $\p = \mathfrak{\mathbb{R}}^{3,1}$, 
    $\p^* = \mathfrak{\mathbb{R}}^{3,1*}$ and $\bp = \bar{\mathbb{R}}^{3,1}$ obey the relations (the isomorphism $\mathrm{End}(\mathbb{R}^n) \simeq \mathbb{R}^n\otimes{\mathbb{R}^{n*}}$ is used):
		\begin{align*}
		\tilde{e}^b (e_a) = \delta_{~a}^b \text{ (dual basis)}, \qquad \qquad \qquad & \bar{e}_a (e_b) = \eta_{ab}, & J_{ab} = e_a \bar{e}_b - e_b \bar{e}_a
	\end{align*}
	and 
	\begin{align*}
		J_{ab} e_c &= (J_{ab})^r_{~c} e_r  = (\eta_{bc}\delta^r_{~a}-\eta_{ac}\delta^r_{~b}) e_r = \eta_{bc} e_a - \eta_{ac} e_b \\[2mm]
		\bar{e}_c J_{ab} &= ( \eta_{ac} \delta^r_b - \eta_{bc} \delta^r_a) \bar{e}_r = \eta_{ac} \bar{e}_b - \eta_{bc} \bar{e}_a \\[2mm]
		[J_{ab},J_{cd}] &= \eta_{bc}J_{ad} + \eta_{ad} J_{bc} + \eta_{bd}J_{ca} + \eta_{ac}J_{db} = \tfrac{1}{2} C^{ef}_{ab,cd} J_{ef}
	\end{align*}
	where the structure constants $C^{fe}_{ab,cd}=-C^{ef}_{ab,cd}$ can be easily computed.
	
	\smallskip
	An element $\varphi = \varphi_\h \oplus \varphi_\m \in \g'= \h  \oplus \m$ admits the following matrix representation: 
	\begin{align}
		\varphi &= 
		\begin{pmatrix}
			\varphi_\h	& k_1 \varphi_\p \\
			k_2 \bar{\varphi}_{\p} & 0
		\end{pmatrix} 
		=
		\begin{pmatrix}
			\dfrac{1}{2}\varphi_\h^{ab} J_{ab}	& k_1 \varphi_\p^a e_a \\
			k_2 \varphi_{\p}^a \bar{e}_a & 0
		\end{pmatrix}
		=
		\dfrac{1}{2} \varphi_\h^{ab}
		\begin{pmatrix}
			J_{ab}	& 0 \\
			0 & 0
		\end{pmatrix}
		\oplus
		\varphi_\m^a
		\begin{pmatrix}
			0	& k_1 e_a \\
			k_2 \bar{e}_a & 0
		\end{pmatrix} \\
		&= \dfrac{1}{2} \varphi_\h^{ab} J_{ab} \oplus \varphi_\m^a M_a
	\end{align}
	where $J_{ab}$ is canonically identified with the matrix representation $J_{ab} = \begin{pmatrix}
        J_{ab} & 0 \\ 0 & 0
    \end{pmatrix}$ of $\h\subset\g'$,  $\varphi^a_\m = \varphi^a_\p$ and $M_a = \begin{pmatrix}
		0	& k_1 e_a \\
		k_2 \bar{e}_a & 0
	\end{pmatrix} = \mu \big( \begin{pmatrix}
		0	& e_a \\
		0 & 0
	\end{pmatrix} \big)$ are the generators of $\m$ constructed out of the canonical basis of $\p = \mathfrak{\mathbb{R}}^{3,1}$ and the parameters of the mutation.
	
	Let $\varphi,\varphi' \in \g'= \h  \oplus \m$,  their Lie bracket is given by the commutator: 
	\begin{align*}
		[\varphi,\varphi'] &= 
		\begin{bmatrix}
			\begin{pmatrix}
				\varphi_\h	& k_1 \varphi_\p \\
				k_2 \bar{\varphi}_{\p} & 0
			\end{pmatrix},
			\begin{pmatrix}
				\varphi'_\h	& k_1 \varphi'_\p \\
				k_2 \bar{\varphi}'_{\p} & 0
			\end{pmatrix}
		\end{bmatrix} 
		=	\begin{pmatrix}
			[\varphi_\h , \varphi'_h] + k_1k_2 (\varphi_\p \bar{\varphi}'_{\p} -  \varphi'_\p \bar{\varphi}_{\p}) & k_1 (\varphi_\h \varphi'_\p - \varphi'_\h \varphi_\p) \\
			k_2(\bar{\varphi}_{\p} \varphi'_\h - \bar{\varphi}'_{\p} \varphi_\h) & 0
		\end{pmatrix} \\
		& = \bigl( \dfrac{1}{8} \varphi_\h^{ab} \varphi_h^{\prime cd}  C_{ab,cd}^{ef} + \dfrac{k_1k_2}{2} ( \varphi_\p^{e} \varphi_\p^{\prime f} - \varphi_\p^{f} \varphi_\p^{\prime e} )  \bigr) J_{ef}
		\oplus 
		\dfrac{1}{2} (  \varphi_\h^{ab} \varphi_\p^{\prime c} - \varphi_\h^{\prime ab} \varphi_\p^{c} ) ( \eta_{bc} \delta^r_a - \eta_{ac} \delta^r_b  ) M_r %\numberthis
	\end{align*}
    where we have used the commutators $[J_{ab},M_c]=\eta_{bc} M_a - \eta_{ac} M_b$ and $[M_a,M_b]=k_1k_2\,J_{ab}$.\footnote{It is interesting to notice that the product of the mutation parameters $k_1$ and $k_2$ can be related to a deformation parameter $\alpha$ as given in \cite[eq.(2.1)]{chirco2025gravityadsyangmillsnatural}.}

 \medskip
	In order to link our notation with the one used in \cite{wise_macdowell-mansouri_2010}, we identify $k_1 = \dfrac{1}{\ell}$ and $k_2 = \dfrac{k}{\ell}$.
    A Cartan connection $\varpi$ associated to this mutated $G/H$ Cartan geometry can be decomposed as:
	\begin{align}
		\varpi = \varpih \oplus \varpim
		=
		\begin{pmatrix}
			\dfrac{1}{2}\varpih^{ab} J_{ab}	& \dfrac{1}{\ell} \varpip^a e_a \\
			\dfrac{k}{\ell} \varpip^a \bar{e}_a & 0
		\end{pmatrix}
		=
		\begin{pmatrix}
			\dfrac{1}{2}A^{ab} J_{ab}	& \dfrac{1}{\ell}\beta^a e_a \\
			\dfrac{k}{\ell}\beta^a \bar{e}_a & 0
		\end{pmatrix}
		= \begin{pmatrix}
			A & \dfrac{1}{\ell}\beta \\
			\dfrac{k}{\ell} \bar{\beta} & 0
		\end{pmatrix}
	\end{align}
	where $\varpih^{ab}=\varpihmod[ab]{\mu} dx^\mu = \Amod[ab]{\mu} dx^\mu $ corresponds
	to the spin connection and $\varpim = \beta^a M_a$ is the $\g'/\h=\m$-valued soldering form (see \cite[Chap.5, §3, Definition 3.1]{sharpe2000differential}) after mutation.
	
	The curvature associated to the covariant derivative and the Cartan connection is:
	\begin{align}
		\bOm & = d \varpi + \dfrac{1}{2} [\varpi,\varpi] = \bOmh \oplus \bOmm
        = \begin{pmatrix}
			R + \dfrac{k}{\ell^2} \beta\wedge \bar{\beta}& 0 \\
			0 & 0
		\end{pmatrix} + \dfrac{1}{\ell} \begin{pmatrix}
			0 & T \\
			k\bar{T} & 0
		\end{pmatrix}  \label{eq:curvMut}
        \\
		& = \bigl(d \varpih + \dfrac{1}{2} [\varpih,\varpih] + \dfrac{1}{2} [\varpim,\varpim]\bigr) 
		\\
		& \qquad \oplus 
		\bigl( d \varpip^r + \dfrac{1}{4} ( \varpihmod[ab]{\mu} \varpipmod[c]{\nu} - \varpihmod[ab]{\nu} \varpipmod[c]{\mu} ) ( \eta_{bc} \delta^r_a - \eta_{ac} \delta^r_b  ) dx^\mu \wedge dx^\nu \bigr) \otimes M_r \notag \\
		& = \bigl( \dfrac{1}{2} \partial_{\mu} \Amod[rs]{\nu} + \dfrac{1}{32}  \Amod[ab]{\mu} \Amod[cd]{\nu} C_{ab,cd}^{rs} 
		+ \dfrac{k}{4} \xi^{rs}_{~~\mu \nu} \bigr) dx^\mu \wedge dx^\nu \otimes J_{rs} \nonumber \\
		& \qquad 
		\oplus 
		\bigl( \partial_{\mu} \betamod[r]{\nu} + \dfrac{1}{2} \Amod[ab]{\mu} \betamod[c]{\nu} ( \eta_{bc} \delta^r_a - \eta_{ac} \delta^r_b ) \bigr) dx^\mu \wedge dx^\nu \otimes M_r
	\end{align}
    with $\bOmm$ the torsion of the Cartan connection $\varpi$ with values in $\g'/\h =\m$. 
        For further use it is worthwhile to notice that the graded bracket
    \begin{align} 
       &\tfrac{1}{2} [\varpim,\varpim] = \dfrac{k}{\ell^2} \begin{pmatrix}
			\beta\wedge \bar{\beta}& 0 \\
			0 & 0
		\end{pmatrix}
        = \begin{pmatrix}
			k\, \xi & 0 \\
			0 & 0
		\end{pmatrix} \in \h \label{eq:mm}
    \end{align}
    with $\xi = \dfrac{1}{\ell^2} \beta \wedge \bar{\beta}= \dfrac{1}{2} \xi^{ab} J_{ab} = \dfrac{1}{4} \xi^{ab}_{~~\mu\nu} dx^\mu \wedge dx^\nu \otimes J_{ab} $, 
	where $ \xi^{ab} = \dfrac{1}{\ell^2} \varpip^{a} \wedge \varpip^{b} $ 
	and $ \xi^{ab}_{~~\mu\nu} = \dfrac{1}{\ell^2} (\varpipmod[a]{\mu} \varpipmod[b]{\nu} - \varpipmod[a]{\nu} \varpipmod[b]{\mu})$ an additional term due to the symmetric Lie algebra structure.
    Let us add that \eqref{eq:curvMut} and \eqref{eq:mm} exemplify \cite[Proposition 6.3, p.218]{sharpe2000differential}.

    \medskip\noindent
	We also identify the spacetime curvature and torsion of Einstein-Cartan gravity:
	\begin{align}
		R &= \bOmh - k \xi = d A + \dfrac{1}{2} [A,A] = dA+A^2 \\
		T &= d\beta + A\wedge\beta = T^a e_a =\bigl( \partial_{\mu} \betamod[r]{\nu} + \dfrac{1}{2} \Amod[ab]{\mu} \betamod[c]{\nu} ( \eta_{bc} \delta^r_a - \eta_{ac} \delta^r_b ) \bigr) dx^\mu \wedge dx^\nu \otimes e_r
	\end{align}
	
	The Bianchi identity $D\bOm = d\bOm + [\varpi,\bOm] =0$ splits as:
		\begin{align}
		\label{New Bianchi}
		(D\bOm)_\h &= dR +[A,R] + \dfrac{k}{\ell^2}  \bigl(\underbrace{d(\beta \wedge \bar{\beta}) + [A,\beta \wedge \bar{\beta}] + \beta \wedge \bar{T} - T\wedge \bar{\beta}}_{=\,0} \bigr)
		= dR +[A,R] = 0 
		\notag \\
		\\[-4mm]
		(D\bOm)_\m &= \dfrac{1}{\ell} (dT + A \wedge T -
		R \wedge \beta)= 0
		\notag
	\end{align}
	with $\bar{T}= d\bar{\beta} + \bar{\beta} \wedge A =T^a\bar{e}_a = \eta(T)$.
	
	\subsection{Hodge $*$-operator}\label{Hodge operator}
	
	The Hodge operator defined on differential forms on the manifold $\mathcal{M}$ requires a metric $g$ on $\mathcal{M}$. According to the reductive Cartan geometry $G/H$ at hand, $\g = \h \oplus \m$,\footnote{In this Section, one denotes $\g$  the mutated algebra $\g'$ defined in \eqref{mutation} and $\p \simeq \m$.} the Killing form $K_\g(\varphi,\varphi') = \Tr(\ad_{\varphi} \circ \ad_{\varphi'})$ on $\g$ splits into two parts: 
	\begin{align}
		K_\h : \h \times \h \rightarrow \mathbb{R} & \text{ and } K_\m : \m \times \m \rightarrow \mathbb{R}
	\end{align}
    where the trace is performed with respect to the generators $J_{ab}, a<b$ and $M_a$.\footnote{In full generality, one has $\Tr(\ad_{\varphi} \circ \ad_{\varphi'})=(m-1)\Tr(\varphi_\h  \varphi'_\h) + 2k_1k_2(m-1) \eta(\varphi_\p , \varphi'_\p)$.}
 
    Let $\kappa$ and $\zeta$ be real scalars.
    In this case we choose $h= \kappa  K_\g$ as the metric on the Lie algebra $\g$ and $g = \zeta \varpim^* h = \kappa \zeta\varpim^* K_\m$. It corresponds to the pullback of the $\m$-part of the Killing form on $\g$ by $\varpim$ the $\m$-part of the Cartan connection up to the factor $\kappa \zeta$. 
	
	For a manifold $\mathcal{M}$ of dimension $m=4$ with Lorentzian Cartan geometry defined by the mutation \eqref{mutation} with $k_1 = 1/\ell$ and $k_2 = k/\ell$: 
    \begin{align} \label{eq:Killing_Lorentz}
		K_\g (\varphi,\varphi') & = K_\h(\varphi_\h , \varphi'_\h) + K_\m(\varphi_\m , \varphi'_\m) = 3 \Tr(\varphi_\h  \varphi'_\h) + \dfrac{6k}{\ell^2} \eta(\varphi_\p , \varphi'_\p) \notag \\
		&= - 3 \eta_{ac} \eta_{bd} \varphi_\h^{ab} \varphi_\h^{\prime cd} + \dfrac{6k}{\ell^2} \eta_{ab} \varphi_\m^a \varphi_\m^{\prime b}
	\end{align}
	with $K_\m(M_a,M_b)=K_{\m,ab} = \dfrac{6k}{\ell^2} \eta_{ab}$ and thus, for any vector fields $X,Y$ on $\mathcal{M}$,
	\[
	g (X,Y) = \kappa \zeta (\varpim^*K_\m) (X,Y) = \dfrac{6k \kappa \zeta }{\ell^2} \eta \bigl( \beta(X) , \beta(Y) \bigr) = \dfrac{6k \kappa \zeta }{\ell^2} \tilde{g} (X,Y)
	\]
    where $\tilde{g} = \beta^* \eta$ corresponds to the metric usually defined in the tetrad formalism.
	
	Let $\omega \in \Omega^{r}(\mathcal{M},\g) $ be an $r$-form on $\mathcal{M}$ with values in $\g$.
	Its local trivialization in a given chart is:
	\begin{align}
		\omega= \dfrac{1}{r!} \omega_{\mu_1\mu_2...\mu_r}dx^{\mu_1}\wedge dx^{\mu_2}\wedge ... \wedge dx^{\mu_r} .
	\end{align}
	    
	Let $\varepsilon_{a_1a_2...a_m}$ be the Levi-Civita symbol in dimension $m$, such that $\varepsilon_{12...m}=1$. From the metric $g$ and its determinant $|g|$, one can define a Hodge star operator $*$ that acts on $\omega$ as:
	\begin{align}
		*\omega = & \dfrac{1}{r!} \sqrt{|g|}\,\omega_{\mu_1...\mu_r} 
		g^{\mu_1 \nu_1}...g^{\mu_r \nu_r} \varepsilon_{\nu_1...\nu_m} dx^{\nu_{r+1}}\wedge ... dx^{\nu_m}
	\end{align}
	In order to stick with the standard literature, the volume form on $\mathcal{M}$ is denoted
	\[
	dvol = \sqrt{|\tilde{g}|} d^m x = \sqrt{|\tilde{g}|} \dfrac{1}{m!}\varepsilon_{\mu_1 \mu_2 ... \mu_m} dx^{\mu_1} \wedge dx^{\mu_2} \wedge ... \wedge dx^{\mu_m} 
	= \dfrac{1}{m!} \varepsilon_{a_1 a_2 ... a_m} \betamod[a_1]{} \wedge \betamod[a_2]{} \wedge ... \wedge \betamod[a_m]{}.
	\]

	\subsection{Associated action}
	
	In the remainder of this article we relate the Levi-Civita symbol to the Levi-Civita tensor according to the relations: \[
 \varepsilon_{abcd}= \varepsilon_{\mu_1 \mu_2 \mu_3 \mu_4} \betainvmod[\mu_1]{a} \betainvmod[\mu_2]{b} \betainvmod[\mu_3]{c} \betainvmod[\mu_4]{d} \quad \text{and} \quad
\varepsilon^{abcd} = \varepsilon^{\mu_1 \mu_2 \mu_3 \mu_4} \betamod[a]{\mu_1} \betamod[b]{\mu_2} \betamod[c]{\mu_3} \betamod[d]{\mu_4}.
 \]

The trace $\Tr$ will be used to compute the Pontrjagin number associated to the  $\g$-valued curvature $\bOm$. Here, $\forall \varphi,\varphi' \in \g=\h\oplus\m$, $\Tr(\varphi\varphi') = - \eta_{ac} \eta_{bd}\varphi_\h^{ab} \varphi_\h^{'cd} + \dfrac{2k}{\ell^2} \eta_{ab} \varphi_\m^a \varphi_\m^{'b} = K_\g (\varphi,\varphi')/3$,
see \eqref{eq:Killing_Lorentz}.
		
\medskip
	By combining linearly the invariant polynomials $\Pf (\dfrac{\bOmh}{2\pi})$ and $\Tr (\bOmh^2)$ (corresponding respectively to deformations of the Euler and Pontrjagin numbers of $R$ with deformation $k\xi$) with the Pontrjagin number $\mathdutchcal{P}$ (see \eqref{Pontrjagin number}) associated to $\bOm$, one can build the action (for $e,r,y$ real numbers\footnote{Complex numbers may be taken for example to describe a complex formulation \cite[p.13]{rezende_4d_2009} 
    of Loop Quantum Gravity (LQG) with a complex Barbero-Immirzi parameter $\gamma = \dfrac{2\alpha_1}{2\alpha_2-\alpha_5} = \dfrac{e}{r+2y}$.}):
	\begin{align*}
		\label{gauge action}
		S_G[\varpi] &= r \mathdutchcal{P}(\overline{\Om}) + \int_\mathcal{M} (e \Pf (\dfrac{\overline{\Om}_\h}{2\pi}) + y \det ( \mathds{1} + \dfrac{\bOmh}{2\pi})) \\
  		&= r \mathdutchcal{P}(\bOm) + \int_\mathcal{M} \bigl(e \Pf (\dfrac{\bOmh}{2\pi}) - \dfrac{y}{8\pi^2} \Tr (\bOmh\bOmh) \bigr)
        = \int_{\mathcal{M}}  \Bigl( e \Pf\left( \tfrac{\bOmh}{2\pi} \right) + r P( \bOm ) - \dfrac{y}{8\pi^2} \Tr (\bOmh\bOmh)\Bigr) \\
		%%%%%%%%%%%%%%%%
		& = \int_{\mathcal{M}} \Bigl(
		\dfrac{ e }{2 (4\pi)^2} 
		\varepsilon_{abcd} \bOmh^{ab} \wedge \bOmh^{cd}  
		- \dfrac{r }{8\pi^2} \Tr (\bOm \wedge \bOm)
		- \dfrac{y }{8\pi^2} \Tr (\bOmh \wedge \bOmh)
		\Bigr) \\
		%%%%%%%%%%%%%%%%
		&= \int_{\mathcal{M}} \Bigl(
		\dfrac{e }{2(4\pi)^2} 
		\varepsilon_{abcd} \bigl(  R^{ab} \wedge R^{cd} 
		+ 2k R^{ab} \wedge \xi^{cd} 
		+ k^2 \xi^{ab} \wedge \xi^{cd}
		\bigr) \\
		& \qquad \quad + \dfrac{1}{8 \pi^2} \bigl( (r+y) (R^{ab} \wedge R_{ab} + 2k  R^{ab} \wedge \xi_{ab}) - \dfrac{2rk}{\ell^2} T^a \wedge T_a
		\bigr)
		\Bigr) \\
		%%%%%%%%%%%%%%%%
		&= \int_{\mathcal{M}} \Bigl(
		\dfrac{e }{2(4\pi)^2} \bigl( 
		\varepsilon_{abcd} R^{ab} \wedge R^{cd} + \dfrac{4k}{\ell^2}  \sqrt{|\tilde{g}|} \Rmod[\mu\nu]{ab} \betainvmod[\mu]{a} \betainvmod[\nu]{b} d^4x + \dfrac{24k^2}{\ell^4} \sqrt{|\tilde{g}|} d^4x
		\bigr) \\
		& \qquad \quad + \dfrac{1}{8 \pi^2} \bigl( (r+y) (R^{ab} \wedge R_{ab} + \dfrac{k}{\ell^2} \sqrt{|\tilde{g}|} \Rmod[\mu\nu]{ab} \betainvmod[\mu]{c} \betainvmod[\nu]{d} \varepsilon^{cd}_{~~ab} d^4x) - \dfrac{2rk}{\ell^2} T^a \wedge T_a
		\bigr)
		\Bigr) \numberthis
	\end{align*}
	
	Let $\gamma^i$, $i=0,1,2,3$ be the Dirac gamma matrices.
	We define a nondegenerate sesquilinear form on $\mathbb{C}^{4}$ : $(\psi,\chi) = \psi^\dagger \gamma^0 \chi = \widebar{\psi} \chi$.
	This form allows us to define a symmetric metric on $\mathbb{C}^{4}$ as $h^\varepsilon(\psi,\chi) = \dfrac{1}{2} \bigl((\psi,\chi) + (\chi,\psi) \bigl)$. 

Using the isomorphism $\Phi_\varepsilon : \mathfrak{so}(n-1,1)\rightarrow \mathfrak{spin}(n-1,1)$ between the Lorentz Lie algebra $\mathfrak{so}(n-1,1)$ and $\mathfrak{spin}(n-1,1)$, \cite[p.192]{gockeler_differential_1987}, whose action on the Lie algebra generators of $\mathfrak{so}(n-1,1)$ is given by:
    \begin{align}
    \label{isomorphism}
		\Phi_\varepsilon (J_{rs}) 
        = \dfrac{1}{4} \bar{e}_a J_{rs} e_b \gamma^a \gamma^b 
        = \dfrac{1}{4} (J_{rs})_{ab} \gamma^a \gamma^b 
        =  \dfrac{1}{2} \eta_{ar} \eta_{bs} \gamma^a \gamma^b \ ,
	\end{align}
 the covariant derivative associated to $\varpih$ on the space of spinors reads $D_\h
	= d + \Phi_\varepsilon (\varpih) 
	= (\partial_\mu + \dfrac{1}{4} \varpihmod[]{ab,\mu}\gamma^a \gamma^b) dx^\mu$.
	One can then construct the following matter action:\footnote{Since the metrics at the level of the Lie algebra and on $\mathcal{M}$ are always defined up to the factors $\kappa$ and $\zeta$, we set $\kappa^2\zeta = \tfrac{\ell^4}{216k^2}$ such that the overall factor is arranged to stick with the Dirac action given in \cite[p.197]{gockeler_differential_1987}.}
    \begin{align}
		\label{Matter action}
		\hskip -3mm
        S_M[A,\beta] & = \int_{\mathcal{M}} h^\varepsilon\bigl( i \psi, h_{\m,ab} \gamma^a \varpim^b \wedge *D_\h\psi\bigl) 
        =\int_{\mathcal{M}} \dfrac{36k^2\kappa\zeta^2}{\ell^4} h^\varepsilon\bigl( i \psi, \eta_{ab} \gamma^a \beta^b \wedge \tilde{*}D_\h\psi\bigl)
        = \mathrm{Re} S_D
	\end{align}
	where $\mathrm{Re} S_D$ corresponds to the real part of the Dirac action as described in \cite[p.197]{gockeler_differential_1987} and $\tilde{*}$ is the Hodge star operator defined from $\tg$.

    \bigskip
	We are now in position to identify the different terms of $S_G[A,\beta]$ constructed above from the invariant polynomials and the Pontryagin characteristic class with the action presented in \cite{rezende_4d_2009}\footnote{This reference is preferably chosen due to the freedom of the parameters whereas the action given in \cite{Freidel:2005ak}[eqs (23-25)] is already more constrained.} which, in our notation, reads:
	\begin{align}
		\label{action Perez}
		S'[\varpi] &= \int_{\mathcal{M}} \Biggl( \biggl(
		\overbrace{ \alpha_1 \ell^2 \underbrace{R^{ab} \wedge \star \xi_{ab}}_{Palatini} + \alpha_2 \ell^2 R^{ab} \wedge \xi_{ab} }^{Holst}
		+ \overbrace{ \alpha_3 R^{ab} \wedge R_{ab} }^{Pontrjagin} \\
		& + \overbrace{ \alpha_4 R^{ab} \wedge \star R_{ab} }^{Euler}
		+ \overbrace{\alpha_5 \bigl( T^a \wedge T_a - \ell^2R^{ab} \wedge \xi_{ab}  \bigr)}^{Nieh-Yan}
		+ \overbrace{\alpha_6 \ell^4 \varepsilon_{abcd} \xi^{ab} \wedge \xi^{cd} }^{Cosmological~constant} 
		\biggr)
		\Biggr) \nonumber
	\end{align}
	where the $\star$ is an internal Hodge star  such that $ \star \beta^{ab} = \dfrac{1}{2} \varepsilon^{ab}_{~~kl}\beta^{kl} $ as in \cite[p.2]{montesinos_self-dual_2001}.
	
	\smallskip
	For $G=ISO(3,1)$ ($k=0$) the action simplifies to the topological terms:
	\begin{align}
		S_G[\varpi] &= \int_{\mathcal{M}} \Bigl(
		\overbrace{\dfrac{e }{2(4\pi)^2} \varepsilon_{abcd} R^{ab} \wedge R^{cd}}^{Euler}
		+ \overbrace{\dfrac{r+y}{8 \pi^2} R^{ab} \wedge R_{ab}}^{Pontrjagin}
		\Bigr) .
		\label{eq:action-ISO}
	\end{align}
	In this case, from the first Bianchi identity $D\bOmh = 0$, and by neglecting boundary terms, the equations of motion for $A$ and $\beta$ require respectively a null spin density and a null energy-momentum tensor as source terms and thus, turn out to be physically irrelevant.
	
	While for general mutation $\mu$ with parameters $k,\ell$, (recall the particular cases $G'=SO(4,1), SO(3,2)$ for $k=-\epsilon= -1,+1$, respectively) we compare $S'$ \eqref{action Perez} and $S_G$ \eqref{gauge action} and identify the 6 coupling constants present in $S'$ via the system:
	\begin{equation}
		\begin{cases}
			\label{system}
			\dfrac{e k}{8\pi^2 \ell^2}  = \alpha_1
			\\
			\dfrac{(r+y) k }{4\pi^2 \ell^2} = (\alpha_2 - \alpha_5)
			\\
			\dfrac{r +y}{8\pi^2} = \alpha_3 
			\\
			\dfrac{e }{(4\pi)^2} = \alpha_4
			\\
			- \dfrac{r k }{4\pi^2 \ell^2} = \alpha_5  
			\\
			\dfrac{e k^2}{32\pi^2\ell^4}  = \alpha_6 
		\end{cases}
		\Leftrightarrow
		\begin{cases}
			\alpha_1 = \dfrac{2k}{\ell^2} \alpha_4 = \dfrac{4\ell^2}{k} \alpha_6  \\
			\alpha_2 - \alpha_5 = \dfrac{2k}{\ell^2} \alpha_3 = - (1+ \dfrac{y}{r}) \alpha_5 \\
			\alpha_1 = \dfrac{e k}{(r+y)\ell^2} \alpha_3 \\
			r +y= \dfrac{4\pi^2\ell^2}{k} (\alpha_2 -\alpha_5)
		\end{cases}
	\end{equation}
	
	The solution of \eqref{system} in terms of the parameters $e$, $r$, $k$, and $\ell$ is: 
	\begin{equation}
		\begin{cases}
			\label{solution 1}
			\alpha_1 = \dfrac{2k}{\ell^2} \alpha_4 = \dfrac{4\ell^2}{k} \alpha_6  \\
			\alpha_2 = \dfrac{ky}{4\pi^2 \ell^2} =\dfrac{2k y}{(r+y)\ell^2} \alpha_3 = - \dfrac{y}{r} \alpha_5 \\
			\alpha_1 = \dfrac{e}{2y} \alpha_2
		\end{cases}
	\end{equation}
	and the resulting deformed topological gauge action reads:
	\begin{align}
		S_G[\varpi] &= r \mathdutchcal{P}(\bOm) + \int_\mathcal{M} \bigl(e  \Pf (\dfrac{\bOmh}{2\pi}) 
        - \dfrac{y}{8\pi^2} \Tr (\bOmh\wedge\bOmh) \bigr) \label{topological action} \\ 
		&= \int_{\mathcal{M}} \Bigl(
		\overbrace{ \dfrac{k }{4\pi^2\ell^2} \big( \dfrac{e}{4} \underbrace{R^{ab} \wedge \beta^c\wedge\beta^d \varepsilon_{abcd}}_{Palatini} + y R^{ab} \wedge \beta_a\wedge\beta_b \big)}^{Holst}
		+ \overbrace{\dfrac{e k^2}{32\pi^2\ell^4}  \beta^a\wedge\beta^b \wedge \beta^c\wedge\beta^d \varepsilon_{abcd}}^{Bare~Cosmological~constant} \nonumber
		\\
		%%%%%%%%
		& \qquad \quad + \overbrace{\dfrac{r +y}{8 \pi^2} R^{ab} \wedge R_{ab}}^{Pontrjagin}
		+ \overbrace{\dfrac{e }{2(4\pi)^2} R^{ab} \wedge R^{cd} \varepsilon_{abcd}}^{Euler} 
		\overbrace{- \dfrac{r k }{4 \pi^2 \ell^2} (T^a \wedge T_a - R^{ab} \wedge \beta_a\wedge\beta_b)}^{Nieh-Yan}
		\Bigr) \nonumber
	\end{align}
As already noticed in \cite{rezende_4d_2009}, there are two contributions to the Barbero-Immirzi parameter $\gamma$, which is expressed here as $\gamma =\dfrac{2\alpha_1}{2\alpha_2-\alpha_5} = \dfrac{e}{r+2y}$. One in front of $R^{ab} \wedge \beta_a\wedge\beta_b$ in the Holst term and the other one in the Nieh-Yan topological term.  

	The equations of motion yield (neglecting boundary terms):
	\begin{align}
		\dfrac{\delta \mathcal{L}_G [A,\beta]}{\delta \beta^c}
		%%%%%%%
		& = \dfrac{ek}{8\pi^2\ell^2} ( R^{ab} + \dfrac{ k }{\ell^2} \beta^a \wedge \beta^b ) \wedge \beta^d \varepsilon_{abcd}
		+ \dfrac{yk}{2\pi^2\ell^2} R^{ab} \wedge \beta_b \eta_{ac}
		%%%%%%%%
		= \tau_c = - \dfrac{\delta \mathcal{L}_M [A,\beta]}{\delta \beta^c} \label{eom frame Lorentz}  \\
		%%%%%%%%%%%%%%%%%%%%%%%%%%%%%%%%%%%%%%%%%%%%%%%%%%%%%%%%%
		\dfrac{\delta \mathcal{L}_G [A,\beta]}{\delta A^{ab}}
		%%%%%%%%
		& = 
		\dfrac{e k}{8\pi^2 \ell^2}T^c \wedge \beta^d \varepsilon_{abcd} 
		+ \dfrac{yk}{2\pi^2\ell^2} T_a \wedge \beta_b
		%%%%%%%%%
		= \dfrac{1}{2} \mathfrak{s}_{ab} = - \dfrac{\delta \mathcal{L}_M [A,\beta]}{\delta A^{ab}} \label{eom spin Lorentz}
	\end{align}
	where $\tau$ and $\mathfrak{s}$ are respectively the energy-momentum tensor and the spin density of matter.
	
	\smallskip
	For $y=0$ and $\alpha_1 = - \dfrac{1}{16 \pi G}$, one identifies\footnote{At this stage a remark is in order. Indeed, the $[\m,\m]$ bracket occurring in the $\h$-part of the curvature, see \eqref{eq:mm} which is related to the mutation parameters $k,\ell$ turns out to be related to the bare cosmological constant $\Lambda_0$. This gives a physical interpretation of the mutation. See Section~\ref{Topological total action} where the topological features of the theory will be linked to the vanishing of the bare cosmological constant $\Lambda_0$.} $ \ell^2 = - \dfrac{3 k}{ \Lambda_0 }$ and $e = \dfrac{8 \pi^2 \alpha_1 \ell^2}{k} = \dfrac{3\pi}{2\Lambda_0 G} $ such that applying the Hodge star $*$ to the first equation \eqref{eom frame Lorentz} gives the Einstein equations with bare cosmological constant $\Lambda_0$ while \eqref{eom spin Lorentz} reduces to the equation relating spin-density to torsion of Einstein-Cartan gravity:
	\begin{align}
    \label{eom lorentz 1}
		G_{kc} + \Lambda_0 \eta_{kc} & = - \dfrac{1}{2\alpha_1} \tau_{ck} \\
		T^c \wedge \beta^d \varepsilon_{abcd} & = \dfrac{1}{2\alpha_1} \mathfrak{s}_{ab}\ \label{eom lorentz 2}.
	\end{align}
    One may notice that the previous identification allows to recover the value (up to the mutation parameter $k$ and sign convention) $ e/(32\pi^2) = -\ell^2/(64\pi Gk)$ of the parameter $\alpha$ of \cite[Sec II]{Miskovic:2009bm} in front of the Euler density built out of $R$, leading to finite Noether charges and a regularized Euclidean action as mentioned in \cite{Aros:1999id,Miskovic:2009bm}.

    \medskip
	Separating the energy-momentum tensor $\tau_{ck}$ as $\tau_{ck}=\tau_{M,ck} -  \rho_{\text{vac}} \eta_{ck}$ with $\tau_{\text{vac},ck}=-\rho_{\text{vac}}\eta_{ck}$ \cite{carroll_cosmological_2001} the part related to the vacuum energy density $\rho_{\text{vac}}$ leads to:
	\begin{align}
		G_{kc} + \Lambda \eta_{kc} & = - \dfrac{1}{2\alpha_1} \tau_{M,ck} \\
		T^c \wedge \beta^d \varepsilon_{abcd} & = \dfrac{1}{2\alpha_1} \mathfrak{s}_{ab}
	\end{align}
	where $\Lambda = \Lambda_0 + \Lambda_\text{vac} = \overbrace{- \dfrac{3k}{\ell^2}}^{\Lambda_0} \overbrace{-\dfrac{\rho_{\text{vac}}}{2\alpha_1}}^{\Lambda_\text{vac}} = -\dfrac{3k}{\ell^2} - \dfrac{4\pi^2\ell^2\rho_{\text{vac}}}{ek} = \Lambda_0 + 8 \pi G \rho_{\text{vac}} $ is the effective cosmological constant, and $\Lambda_\text{vac}$ is the vacuum contribution.
	
	Let $\Lambda_{exp}$ be the measured value of the effective cosmological constant $\Lambda$. Several situations may occur depending on the theoretical predictions for $\rho_\text{vac}$.
	
	\begin{itemize}
		\item    
		If $\Lambda_\text{vac}>\Lambda_{exp}$, then, agreement with experiment could be attained by fixing the appropriate value of $\Lambda_0=-\dfrac{3k}{\ell^2}$ with $k>0$, in particular, for $k=1$ $\Leftrightarrow G=SO(3,2)$ (AdS spacetime), thus potentially allowing the use of AdS/CFT correspondence \cite{aharony_n_2008}.
		\item
		If $\Lambda_{\text{vac}}<\Lambda_{exp}$, agreement with $\Lambda_{exp}$ would require a mutation parameter $k<0$. A particular case is $k = - 1$ $\Leftrightarrow G=SO(4,1)$ (dS spacetime). Especially,  $\rho_{\text{vac}}=0$ whenever one imposes normal ordering as mentioned in \cite{copeland2006dynamics}.
		\item
		The last case is $\Lambda_{\text{vac}}=\Lambda_{exp}$ this either implies $k=0 \Leftrightarrow G=ISO(3,1)$, which corresponds to Minkowski spacetime or $\Lambda_0=-\dfrac{3k}{\ell^2} \to 0$. 
        In this limit, the contribution \eqref{eq:mm} to the $\h$-part of the Cartan curvature becomes small and the theory can be considered as topological. This point will be discussed in detail in Section~\ref{Topological total action}.
    
        However, for $k=0$ we do not retrieve Einstein's equations since the action consists solely of the total derivatives corresponding to both the Euler and Pontrjagin densities of the curvature $R$.
		
	\end{itemize}
	
	\smallskip
	\begin{remark}
		The equations of motion are the same as those of Einstein-Cartan gravity modulo the boundary terms if $y=0$. Even if they do not affect the classical equations of motion, the Euler and Pontrjagin densities of $R$ as well as the Nieh-Yan term (which are not present in usual Einstein-Cartan gravity) could have non trivial effects in a path integral formulation of quantum gravity when summing over all possible configurations as stated in \cite[p.2]{rezende_4d_2009}. The Nieh-Yan term also gives a non trivial contribution to the Barbero-Immirzi parameter $\gamma = \dfrac{e}{r+2y}$ via the parameter $r$.
	\end{remark}

	Using the same relation as in \cite{rezende_4d_2009} we can identify the Barbero-Immirzi parameter $\gamma =\dfrac{2\alpha_1}{2\alpha_2-\alpha_5} = \dfrac{e}{r+2y}$, with our parameters. By replacing $e$, $y$ and $\ell$ in the action we obtain: 
	\begin{align*}
		\label{topological action mutation}
		S_G[\varpi] &= r \mathdutchcal{P}(\overline{\Om}) + \int_\mathcal{M} (\dfrac{3\pi}{2\Lambda_0 G} \Pf (\dfrac{\overline{\Om}_\h}{2\pi}) + (\dfrac{3\pi}{4\Lambda_0 G\gamma} -\dfrac{r}{2}) \det ( \mathds{1} + \dfrac{\bOmh}{2\pi})) \numberthis \\
        %%%%%%%%%%%%%%%%%%%%%%%%%%%%%%%%%%%%%%%%
		&= \int_{\mathcal{M}} \dfrac{\alpha_1}{2} \biggl(
		\overbrace{ R^{ab} \wedge \beta^c\wedge\beta^d \varepsilon_{abcd} + (\dfrac{2}{\gamma} - \dfrac{4 \Lambda_0 G r}{3\pi}) R^{ab} \wedge \beta_a\wedge\beta_b }^{Holst}
		\overbrace{-\dfrac{\Lambda_0}{6}  \beta^a \wedge \beta^b \wedge \beta^c \wedge \beta^d \varepsilon_{abcd}}^{Cosmological~constant} 
		\\
		& \qquad \qquad \overbrace{- \dfrac{3}{2\Lambda_0} R^{ab} \wedge R^{cd} \varepsilon_{abcd} }^{Euler}
		\overbrace{
        - (\dfrac{2Gr}{\pi} + \dfrac{3}{\Lambda_0 \gamma}) R^{ab} \wedge R_{ab} }^{Pontrjagin}
		\overbrace{-\dfrac{8\Lambda_0Gr}{3\pi} \bigl( T^a \wedge T_a - R^{ab} \wedge \beta_a \wedge \beta_b \bigr)}^{Nieh-Yan}
		\biggr) .
	\end{align*}
	Thus, by building an action from a linear combination of these invariant polynomials and the Pontrjagin number associated to a manifold with Lorentzian Cartan geometry we retrieve the Holst action with a bare cosmological constant $\Lambda_0$ as well as the Euler, Pontrjagin and Nieh-Yan densities of the curvature $R$ and torsion $T$. According to  \cite{date_topological_2009}, these three last terms should be present in any quantum theory of gravity.
    
	The way \eqref{topological action} is built from invariant polynomials in the curvature with mutated geometries may be of use to obtain a generalization of the action described in the context of qeGR in \citep{cattaneo2024gravity} that would incorporate all the topological terms appearing in \eqref{topological action} and \citep{rezende_4d_2009}.
 	The main difference compared to the literature is that in our case the coupling constants of the Euler and Pontrjagin densities are determined by the bare cosmological constant and by, respectively, the coupling constant of the Palatini action and the Barbero-Immirzi (B-I) parameter as well as $r$. This is due to how we constructed the gauge action from topological arguments. Moreover, notice that one may identify the coefficient $\alpha_3 = \dfrac{r+y}{8\pi^2}$ 
    related to the Pontrjagin density built from $R$ with the parameter $\beta$ of \cite[eq.(26)]{Miskovic:2009bm}. In this case one would then be able to identify (up to a sign convention) the last remaining physical parameter:
    \begin{align}
    r= - \dfrac{3\pi}{2\Lambda_0 G} (\dfrac{1}{\gamma}\pm 1) .
    \end{align} 

    These dependencies on $G, \gamma$, $\Lambda_0$ and $r$ are interesting since a quantum version of the theory may yield effects sensitive to the bare cosmological constant $\Lambda_0$ in particular.
	We believe that probing the bare cosmological constant in this way could be highly beneficial in constraining the vacuum contribution to the cosmological constant $\Lambda_\text{vac}$ since the two cases $\Lambda_\text{vac}=0$ and $\Lambda = \Lambda_\text{vac}=8 \pi G \rho_{\text{vac}}$ bring about huge differences in the numerical value of $\Lambda_0$. 
 
	\smallskip
    One can also see that the coupling constant of the Euler density $-\dfrac{3\alpha_1}{4\Lambda_0} R^{ab} \wedge R^{cd} \varepsilon_{abcd}$ obtained via the equation of motion \eqref{eom frame Lorentz} for the tetrad $\beta$ matches perfectly with what has been described in \cite[section III]{alexander_zero-parameter_2019} such that the gauge action $S_G$ \eqref{topological action mutation} is invariant under the "duality symmetry" consisting of the exchange $R^{ab} \leftrightarrow \dfrac{\Lambda_0}{3} \beta^a \wedge \beta^b$. It is worthwhile to notice that since by construction $S_G$ is quadratic in the curvature, $S_G$ turns out to be invariant under the (more general) duality symmetry $\bOm \leftrightarrow - \bOm$ which splits as: 
    \begin{equation}
		\begin{cases}
			\label{solution 2}
			\bOmh \leftrightarrow - \bOmh \Rightarrow R \leftrightarrow \dfrac{\Lambda_0}{3} \beta \wedge \bar{\beta} \\
			\bOmm \leftrightarrow - \bOmm \Rightarrow T/\ell \leftrightarrow - T/\ell\ .
		\end{cases}
	\end{equation}
   and reduces to the case of \cite{alexander_zero-parameter_2019} upon solely considering the $\h$-sector of the curvature. Remark that $S_G$ is invariant under each of these transformations independently.

   \smallskip
   In addition, a direct comparison with the two perturbative expansion parameters $\alpha$,$\beta$ given in \cite{Freidel:2005ak} yields:
    \begin{align*}
        \alpha & = \dfrac{G \Lambda_0}{3(1-\gamma^2)} = -\dfrac{r+2y}{e} \beta = -\dfrac{\pi}{2e} \dfrac{(r+2y)^2}{(r+2y)^2-e^2} \\
        \beta & = \dfrac{\gamma G \Lambda_0}{3(1-\gamma^2)} = \dfrac{\pi}{2} \dfrac{r+2y}{(r+2y)^2-e^2}.
    \end{align*}
	As for $k$, the case of a dynamical $\ell$ is studied in \cite{Thibaut:2025lwx}.

\smallskip
	We can now consider $ e = \dfrac{8 \pi^2 \alpha_1 \ell^2}{k} = \dfrac{3\pi}{2\Lambda_0 G} $ and $ \ell^2 = - \dfrac{3 k}{ \Lambda_0 } $ as fixed.
	Upon identifying the B-I parameter as $\gamma = \dfrac{e}{r+2y}$, only two degrees of freedom corresponding to the relative value of the parameters $r$ and $y$ and to $|k|$ (the absolute value of $k$) are left.
	
	\smallskip
	\begin{remark}
		The Pontrjagin density and other terms described in \eqref{topological action mutation} correspond to the Chern-Simons modified gravity \cite{alexander_chern-simons_2009} terms with a constant scalar field $\vartheta = -\big(\dfrac{2Gr}{\pi} + \dfrac{3}{\Lambda_0 \gamma}\big) \dfrac{\alpha_1}{2}$ and potential term for $\vartheta$, ${\displaystyle S_{\vartheta} = \int_{\mathcal{M}} V(\vartheta)}$, with:  
		\begin{align}
			V(\vartheta) &=  
			(\dfrac{1}{\gamma} - \dfrac{ 2\Lambda_0 G r}{3\pi}) \alpha_1 R^{ab} \wedge \beta_a\wedge\beta_b 
			\overbrace{- \dfrac{\Lambda_0\alpha_1}{12} \varepsilon_{abcd} \beta^a \wedge \beta^b \wedge \beta^c \wedge \beta^d }^{Cosmological~constant} \\
			&\qquad
			\overbrace{- \dfrac{3\alpha_1}{4\Lambda_0} R^{ab} \wedge R^{cd} \varepsilon_{abcd} }^{Euler}
			\overbrace{-\dfrac{4\Lambda_0G\alpha_1r}{3\pi} \bigl( T^a \wedge T_a - R^{ab} \wedge \beta_a \wedge \beta_b \bigr)}^{Nieh-Yan}
			. \nonumber
		\end{align}
		The Chern-Simons term also occurs in Loop Quantum Gravity (LQG) and String theory as mentionned in \cite{alexander_chern-simons_2009}. The fact that $\vartheta$ is a constant in our formalism leads to the speed of gravitational waves being the same as the speed of light as discussed in \cite{diaz2023gauss,daniel2024gravitational} and is consistent with the gravitational wave GW170817 event.
		
	\end{remark}
	
	Let us collect some of the results in the table below:
    
	\begin{align} \label{eq:table} 
		\left\{
		\begin{array}{ll}
			k = - 1  &\Leftrightarrow G = SO(4,1) \text{ for } \Lambda_0 > 0 \\[2mm]
			\left\{
			\begin{array}{ll}
				k
				= \pm 1 \Leftrightarrow G=SO(4,1),SO(3,2) \\
				\ell \rightarrow \infty \\
				r=-2y=0
			\end{array}
			\right\}
			\Leftrightarrow
			\left\{
			\begin{array}{ll}
				\Lambda_0 \rightarrow 0 \\
				\gamma \rightarrow \infty
			\end{array}
			\right\}  &\simeq \text{Einstein-Cartan} \ \\
			k = 1  &\Leftrightarrow G = SO(3,2) \text{ for } \Lambda_0 < 0
		\end{array}
		\right.
	\end{align}

Taking $ry \neq 0$ and $r \neq - 2y$ in the action \eqref{topological action} secures the presence of the Barbero-Immirzi part of the Holst action, together with the Pontrjagin density of $R$, the Nieh-Yan term and a non diverging Barbero-Immirzi parameter.
	
	\section{Invariant polynomials and action for the Möbius group $G= SO(4,2)/\{\pm I\}$}
		\label{section mobius}
	
	Another interesting example of geometry is given by the Möbius group $G=SO(4,2)/\{\pm I\}$ defined by the relation $g^T N g = N$, $g \in G$, with 
	\begin{align*}
		N = \begin{pmatrix}
			0 & 0 & -1 \\
			0 & \eta & 0 \\
			-1 & 0 & 0
		\end{pmatrix} 
	\end{align*}  
	and $K= \overbrace{CO(3,1)}^H \ltimes \mathbb{R}^{3,1*} $ the maximal normal subgroup of $G$. A group element of K admits the matrix representation: 
		\begin{align*}
		K \ni \overbrace{\begin{pmatrix}
				\zeta & 0 & 0 \\
				0 & \Lambda & 0 \\
				0 & 0 & \zeta^{-1}
		\end{pmatrix} }^{\in H = CO(3,1) }
		\overbrace{\begin{pmatrix}
				1 & \bar{u} & \tfrac{1}{2} \bar{u}u \\
				0 & \mathds{1} & u \\
				0 & 0 & 1
		\end{pmatrix}}^{\in \mathbb{R}^{3,1*}}
		\text{ where } \begin{array}{ll}
        \Lambda \in SO(3,1), \Lambda^{\text{T}} \eta \Lambda =\eta \; ; \; \zeta\Lambda \in CO(3,1), \\
    u\in \mathbb{R}^{3,1} \; ; \; \bar{u} =  u^{\text{T}}\eta \in {\mathbb{R}^{3,1}}^* \text{ and } \zeta \in \mathbb{R}^+
        \end{array}
	\end{align*}
    where $ \mathbb{R}^+$ denotes the multiplicative group of positive real numbers.
    
	The Lie algebra of $G$ is $\g = \mathfrak{so}(4,2) = \g_{-1} \oplus \g_0 \oplus \g_{1} $ with $\g_{-1} \simeq \mathbb{R}^{3,1} $, $\g_0 \simeq \mathfrak{co}(3,1) = \mathfrak{so}(3,1) \oplus \mathbb{R}$ and $ \g_1 \simeq \mathbb{R}^{3,1*}$. The Lie algebra $\g$ turns out to be a graded Lie algebra \cite{kobayashi_transformation_1995}, namely,
	\begin{align}
		& [ \g_0 , \g_0 ] \subset \g_0, \qquad [ \g_0 , \g_{-1} ] \subset \g_{-1},  \qquad [ \g_1 , \g_0 ] \subset \g_1, \qquad [\g_{-1},\g_1] \subset \g_0 \notag \\[-4mm]
		& \label{eq:gradedLa} \\
		& \qquad \qquad \qquad \qquad \qquad \; \; [ \g_{-1} , \g_{-1} ] = [ \g_{1} , \g_{1} ] = 0. \notag
	\end{align} 
	The generators of the above graded algebra are the same as the ones described previously in \eqref{generators} except for $\g_0 \simeq \mathfrak{co}(3,1)$ which is generated by the $ \{J_{ab} \} $'s and $\mathds{1}_4$.
	An element of $\mathfrak{co}(3,1)\simeq \g_0$ is $\gamma_{\mathfrak{co}} = \frac{1}{2}\gamma_{\mathfrak{so}}^{ab}J_{ab} - z \mathds{1}_4 $ with $\gamma_{\mathfrak{so}} \in \mathfrak{so}(3,,1)$ and $z\in \mathbb{R}$.
	
	\subsection{The Möbius Lie algebra as a symmetric Lie algebra \label{subsect:Moebius}} 

Let us consider $H=CO(3,1)$ as structure group for the Cartan geometry to be discussed in this part. This implies a geometry that differs from the whole isotropic Möbius group as introduced in \cite[Lemma 1.8,p.269]{sharpe2000differential}). It rather corresponds to $SO(3,1)\times \mathbb{R}^+$ (Lorentz$\times$Weyl) the isotropic part of the Weyl model which can be embedded into the Möbius model, see \cite[middle p.277]{sharpe2000differential}.

	The following decomposition of the graded Lie algebra
    \begin{align} \label{eq:graded-equiv}
    \g= \overbrace{\g_0}^{\h}  \oplus \overbrace{\g_{1} \oplus \g_{-1}}^{\m} \simeq \overbrace{\mathfrak{co}(3,1)}^{\mathfrak{so}(3,1) \oplus \mathbb{R}} \oplus {\mathbb{R}^{3,1}}^* \oplus \mathbb{R}^{3,1}
    \end{align}
    of the Möbius group is reductive with respect to $H$ and is also a symmetric Lie algebra $\g = \mathfrak{co}(3,1) \oplus \m$  where $\m= \g_{1} \oplus \g_{-1} \simeq {\mathbb{R}^{3,1}}^* \oplus \mathbb{R}^{3,1}$. The commutation relations (see \eqref{eq:gradedLa} and \cite[Example 5.2 and Chapter XI]{kobayashi_foundations_1996}) are given by
	\begin{align}
		[ \g_0 , \g_0 ] \subset \g_0, \qquad [ \g_0 , \m ] \subset \m,  \qquad [ \m,\m] \subset \g_0\ 
	\end{align} 
    making $\g=\g_0\oplus \m$ a symmetric Lie algebra.
    
According to the graded Lie algebra decomposition of $\g$, let  
    \[
    \varphi_\g = \varphi_1\oplus \varphi_0\oplus \varphi_{-1} \in \g= \g_{1}\, \oplus\, \g_0\, \oplus\, \g_{-1}
    \] 
for which we adopt the matrix presentation\footnote{Accordingly, the Killing form on $\g = \h\oplus\m$ is computed to be 
   \[
   K_\g(\varphi,\varphi') = \Tr(\ad_{\varphi}\circ \ad_{\varphi'}) = n \Tr(cc')+ 2n\big( zz'+\eta(b,a')+\eta(a,b')\big)
   \]
   where $n$ corresponds to the one coming from $\dim \mathfrak{so}(n-1,1) = n(n-1)/2$.}
	\begin{align}\label{eq:graded-decomp}
		\varphi_\g = \varphi_1\oplus \varphi_0\oplus \varphi_{-1} = \begin{pmatrix}
			z & \bar{a}	& 0 \\
			b & c & a \\
			0 & \bar{b} & -z
		\end{pmatrix}
		=
		\begin{pmatrix}
			0 & \bar{a} & 0 \\
			0 & 0 & a \\
			0 & 0 & 0
		\end{pmatrix}
		+
		\begin{pmatrix}
			z & 0 & 0 \\
			0 & c & 0 \\
			0 & 0 & -z
		\end{pmatrix}
		+
		\begin{pmatrix}
			0 & 0 & 0 \\
			b & 0 & 0 \\
			0 & \bar{b} & 0
		\end{pmatrix}
	\end{align}
with $a,b \in \mathbb{R}^{3,1}$, $c \in \mathfrak{so}(3,1)$, $z \in \mathbb{R}$ and $\bar{a}= a^T \eta, \bar{b}= b^T \eta \in \mathbb{R}^{*3,1}$. It is worthwhile to notice that $E = \begin{pmatrix}
    1 & 0 & 0\\ 0 & 0 & 0 \\ 0 & 0 & -1
\end{pmatrix}\in \g_0$ the generator for dilation defines the grading of $\g$ by $\g_j = \{\varphi\in\g \text{ s.t. } [E,\varphi] = j \varphi \}, j=0,\pm 1$.
Moreover, through the isomorphism $\g\simeq {\mathbb{R}^{3,1}}^* \oplus \mathfrak{co}(3,1) \oplus \mathbb{R}^{3,1}$, one can also write \cite{kobayashi_transformation_1995}
\[
\varphi_\g = \varphi_1\oplus \varphi_{\mathfrak{co}} \oplus \varphi_{-1} = \varphi_1\oplus (\varphi_{\mathfrak{so}} - \varphi_{\mathbb{R}}\mathds{1}_4) \oplus \varphi_{-1} =\bar{a} \oplus (c - z \mathds{1}_4) \oplus b.
\]

    Let us now consider the mutation\footnote{See Definition~\ref{defmutation} given above.} $\g'=\mu(\g)=\mu(\h\oplus\m) = \h\oplus\mu(\m)=\h\oplus\m'$ given by:
\begin{align}
		\label{mutation Möbius}
		\varphi_\g =
		\begin{pmatrix}
            z & \bar{a} & 0 \\
			b  & c	& a \\
			0 & \bar{b}  & -z 
		\end{pmatrix}
		\mapsto
		\varphi_{\g'} = \mu(\varphi_\g)=
		\begin{pmatrix}
			z & \bar{\gamma}_2 \bar{a} & 0 \\
			\gamma_1 b   & c & \gamma_2 a \\
			0 & \bar{\gamma}_1 \bar{b}  & -z 
		\end{pmatrix} \in \g'
\end{align}
where the scalar mutation parameters $\gamma_1,\bar{\gamma}_1\gamma_2,\bar{\gamma}_2$ are arranged such that $\mu$ is a linear isomorphism, in particular, $\gamma_1 \bar{\gamma}_2= \bar{\gamma}_1\gamma_2$. Thus, it is merely a change of scale within the symmetric Lie algebra which preserves the splitting $\m=\g_{-1}\oplus\g_1$ according to the graded structure.

\smallskip
Let us choose $\gamma_1=\bar{\gamma}_1=\dfrac{1}{\ell}$, $\gamma_2=\bar{\gamma}_2 = \dfrac{1}{\ell'}$ (where $\ell$ and $\ell'$ are non-null real numbers)\footnote{In that case, the Killing form on $\g'=\h\oplus\m'$ can be easily computed by performing the scale transformations $b\mapsto \gamma_1 b$ and $a \mapsto \gamma_2 a$ into the expression for $K_\g$, so that
\[
K_{\g'}(\varphi,\varphi') = n \Tr(cc')+ 2n\Bigl( zz'+ \gamma_1\gamma_2\bigl(\eta(b,a')+\eta(a,b')\bigr) \Bigr).
\]}.

    The Lie bracket on $\g'$ is thus given by the commutator:
	\begin{align*}
		[\varphi,\varphi'] & =
		%%%%%%%%%%%%%%%
		\begin{bmatrix}
			\begin{pmatrix}
				z & \dfrac{1}{\ell'} \bar{a}	& 0 \\
				\dfrac{1}{\ell} b & c & \dfrac{1}{\ell'} a \\
				0 & \dfrac{1}{\ell} \bar{b} & -z
			\end{pmatrix}
			, %%%%%%%%%%%%%%%
			\begin{pmatrix}
				z' & \dfrac{1}{\ell'} \bar{a}'	& 0 \\
				\dfrac{1}{\ell} b' & c' & \dfrac{1}{\ell'} a' \\
				0 & \dfrac{1}{\ell} \bar{b}' & -z'
			\end{pmatrix}
		\end{bmatrix} \numberthis \\
		%%%%%%%%%%%%%%%%%%%%%
		& =
		\begin{pmatrix}
			\dfrac{1}{\ell\ell'} (\bar{a} b' - \bar{a}' b)  & \dfrac{1}{\ell'} (z\bar{a}' - z' \bar{a}	+ \bar{a} c' - \bar{a}' c) & \dfrac{1}{\ell'^2} (\bar{a} a' - \bar{a}' a) \\
			%%%%%%%%%%%%%%%%%%%%
			\dfrac{1}{\ell} (z' b - z b' + c b' - c' b) & [c,c'] + \dfrac{1}{\ell\ell'} (b \bar{a}' - b' \bar{a} + a \bar{b}' - a' \bar{b}) & \dfrac{1}{\ell'} (c a' - c' a - z' a + z a') \\
			%%%%%%%%%%%%%%%%%%%%
			\dfrac{1}{\ell^2} (\bar{b} b' - \bar{b}' b) & \dfrac{1}{\ell} (\bar{b} c' - \bar{b}' c - z \bar{b}' + z' \bar{b}) & \dfrac{1}{\ell\ell'} (\bar{b} a' - \bar{b}' a)
		\end{pmatrix} \\
		& = 
		\begin{pmatrix}
			\dfrac{1}{\ell\ell'} (\bar{a} b' - \bar{a}' b)  & \dfrac{1}{\ell'} (z\bar{a}' - z' \bar{a}	+ \bar{a} c' - \bar{a}' c) & 0 \\
			%%%%%%%%%%%%%%%%%%%%
			\dfrac{1}{\ell} (z' b - z b' + c b' - c' b) & [c,c'] + \dfrac{1}{\ell\ell'} (b \bar{a}' - b' \bar{a} + a \bar{b}' - a' \bar{b}) & \dfrac{1}{\ell'} (c a' - c' a - z' a + z a') \\
			%%%%%%%%%%%%%%%%%%%%
			0 & \dfrac{1}{\ell} (\bar{b} c' - \bar{b}' c - z \bar{b}' + z' \bar{b}) & \dfrac{1}{\ell\ell'} (\bar{b} a' - \bar{b}' a)
		\end{pmatrix}
	\end{align*}
	
	A Cartan connection $\varpi$ with this $G/H$ Cartan geometry can accordingly be decomposed as:
	\begin{align}
		\varpi = \varpih \oplus \varpim = \varpi_0 \oplus \overbrace{\varpi_{1}\oplus \varpi_{-1}}^{\varpim}
		=& \begin{pmatrix}
			\lambda & 0 & 0 \\
			0 & A & 0 \\
			0 & 0 & - \lambda
		\end{pmatrix}
		+
        \varpi_{1}^a
		\begin{pmatrix}
			0 & \dfrac{1}{\ell'} \bar{e}_a & 0 \\
			0 & 0 & \dfrac{1}{\ell'} e_a \\
			0 & 0 & 0
		\end{pmatrix}
		+
        \varpi_{-1}^a
		\begin{pmatrix}
			0 & 0 & 0 \\
			\dfrac{1}{\ell} e_a & 0 & 0 \\
			0 & \dfrac{1}{\ell} \bar{e}_a & 0
		\end{pmatrix} \nonumber \\
		=&
		\begin{pmatrix}
			\lambda & \dfrac{1}{\ell'} \bar{\alpha} & 0 \\
			\dfrac{1}{\ell} \beta & A & \dfrac{1}{\ell'} \alpha \\
			0 & \dfrac{1}{\ell} \bar{\beta} & - \lambda
		\end{pmatrix}
	\end{align}
	where $\varpim$ is the soldering form, $A= \dfrac{1}{2} A^{ab} J_{ab} =\dfrac{1}{2} A^{ab}_\mu dx^\mu \otimes J_{ab} $ corresponds to the spin connection and 	$\beta = \beta^a e_a = \beta^a_\mu dx^\mu \otimes e_a$ is the tetrad used to define the metric of the tetrad formalism $\tilde{g}= \eta(\beta,\beta)$. 
	$\alpha = \alpha^a e_a = \alpha^a_\mu dx^\mu \otimes e_a$ can be interpreted as a secondary tetrad.
	Finally, $\lambda = \lambda_\mu dx^\mu$ is the $\mathbb{R}$-valued (dilations) part of the connection.
	Also, $\bar{\alpha} = \alpha^a \bar{e}_a = \alpha^a_\mu dx^\mu \otimes \bar{e}_a$ and $\bar{\beta} = \beta^a \bar{e}_a = \beta^a_\mu dx^\mu \otimes \bar{e}_a $.

\smallskip
    \begin{remark}
              The geometry modeled on the present $8$-dimensional quotient $G/H$ describes in general an $8$-dimensional manifold $\widetilde{\mathcal{M}}$ invariant under Lorentz transformations (elements in $SO(3,1)$) and dilations (elements of the multiplicative group $\mathbb{R}^+$) with soldering form $\varpim$. In the next section, we will see how we can obtain a description of a $4$-dimensional submanifold $\mathcal{M}$ of $\widetilde{\mathcal{M}}$ by imposing a condition on the soldering form (the $\m$-valued part of the Cartan connection).
    \end{remark}
 
	The curvature $\bOm = d\varpi + \dfrac{1}{2} [\varpi,\varpi] $ of $\varpi$ then reads:
	\begin{align}
		\bOm & = \bOmh \oplus \bOmm =
		\begin{pmatrix}
			f & \dfrac{1}{\ell'}\bar{\Pi} & 0 \\
			\dfrac{1}{\ell}\Theta & F & \dfrac{1}{\ell'}\Pi \\
			0 & \dfrac{1}{\ell}\bar{\Theta} & - f
		\end{pmatrix}
		=
		\begin{pmatrix}
			f & 0 & 0 \\
			0 & F & 0 \\
			0 & 0 & - f
		\end{pmatrix} 
		+
		\begin{pmatrix}
			0 & \dfrac{1}{\ell'}\bar{\Pi} & 0 \\
			\dfrac{1}{\ell}\Theta & 0 & \dfrac{1}{\ell'}\Pi  \\
			0 & \dfrac{1}{\ell}\bar{\Theta} & 0
		\end{pmatrix}
		\\[4mm]
		& =
		\begin{pmatrix}
			d\lambda +\dfrac{1}{\ell \ell'}  \bar{\alpha} \wedge \beta  
			& \dfrac{1}{\ell'} (d\bar{\alpha} +\lambda \wedge \bar{\alpha} + \bar{\alpha} \wedge A) & 0 \\
			\dfrac{1}{\ell} (d\beta - \lambda \wedge \beta + A \wedge \beta) & dA + \dfrac{1}{2} [A,A] + \dfrac{1}{\ell \ell'} (\beta \wedge \bar{\alpha} + \alpha \wedge \bar{\beta}) & \dfrac{1}{\ell'} (d \alpha + A \wedge \alpha + \lambda \wedge \alpha) \\
			0 & \dfrac{1}{\ell} (d\bar{\beta} + \bar{\beta} \wedge A - \lambda \wedge \bar{\beta}) & - d\lambda + \dfrac{1}{\ell \ell'} \bar{\beta} \wedge \alpha
		\end{pmatrix} \nonumber
	\end{align}
	with:
 	\begin{align*}
		& \h\text{-part}
        \begin{cases}
        \mathfrak{so}\text{-part:}\quad F = R + \dfrac{1}{\ell \ell'} \phi 
		\quad \text{with}~ \phi= \beta \wedge \bar{\alpha} + \alpha \wedge \bar{\beta} \\
        \\[-4mm]
		\text{dilation:} \quad f = d\lambda + \dfrac{1}{\ell \ell'} \bar{\alpha} \wedge \beta 
          \end{cases}
		\\
		& \numberthis \label{eq:courbures} \\[-2mm]
		& \m\text{-part}
        \begin{cases}
            \Pi = d\alpha + A \wedge \alpha + \lambda \wedge \alpha, \qquad \bar{\Pi} = d\bar{\alpha} + \bar{\alpha} \wedge A - \bar{\alpha} \wedge \lambda \\
            \\[-4mm]
		 \Theta = d\beta + A \wedge \beta - \lambda \wedge \beta =  T - \lambda \wedge \beta, \qquad \bar{\Theta} = \bar{T} + \bar{\beta}\wedge\lambda
        \end{cases}
         	\end{align*}
(see {\em e.g.} \cite{Francois:2015oca}).
	For further use, it is worthwhile to notice that the $[\m,\m]\subset\h$ contribution to $\bOmh$ is given by
    \begin{align}
       \tfrac{1}{2} [\varpim,\varpim] = 
       \dfrac{1}{\ell \ell'}  \begin{pmatrix}
			\bar{\alpha} \wedge \beta  
			& 0 & 0 \\
			0 &  \beta \wedge \bar{\alpha} + \alpha \wedge \bar{\beta} & 0 \\
			0 & 0 & \bar{\beta} \wedge \alpha
		\end{pmatrix}.
        \label{eq:mm-moebius}
    \end{align}
	
	The Bianchi identities are given by $D\bOm = d\bOm + [\varpi, \bOm] = 0$ and, on account of the identification \eqref{eq:graded-equiv}, splits into the identities for each sector
\begin{align*}
		& \h\text{-part}
        \begin{cases}
        \mathfrak{so}\text{-part:}\quad dR + [A,R] + \frac{1}{\ell\ell'} \big( d\phi + [A,\phi] + \beta \wedge \bar{\Pi} - \Pi \wedge \bar{\beta} + \alpha \wedge \bar{\Theta} - \Theta \wedge \bar{\alpha} \big) = 0 \\
        \\[-4mm]
        \text{dilation:} \quad d(\bar{\alpha}\wedge\beta) +\bar{\alpha}\wedge\Theta - \bar{\Pi}\wedge\beta = 0 
          \end{cases}
		\\[2mm]
		& \m\text{-part}
        \begin{cases}
            d\Theta + (A-\lambda)\wedge\Theta + (d\lambda - R)\wedge\beta =0 \\
            \\[-4mm]
		 d\Pi + (A+\lambda)\wedge \Pi - (d\lambda - R)\wedge\alpha = 0.
        \end{cases}
         	\end{align*}      
	
	\smallskip
	\begin{remark}
		If we choose a Cartan geometry such that $\Theta= 0$, $f = 0$ and Ric$(F)=0$,  \cite[Chapter 7, Def 2.7]{sharpe2000differential} then, as described in \cite[p.52]{francois_reduction_2014-1}, $\bar{\Pi}$ and $\bar{\alpha}$ become, respectively, the Cotton and Schouten tensors, while $F$ would correspond to the Weyl tensor. These tensors constitute the building blocks of conformal gravity.
	\end{remark}

	\subsection{Characteristic poylynomials for $G = SO(4,2)/ \{\pm I\}$ and associated action}
	
	Recall the Levi-Civita symbols and tensors are still linked via the relations: 
	\[
	\varepsilon_{abcd} = \varepsilon_{\mu_1 \mu_2 \mu_3 \mu_4} \betainvmod[\mu_1]{a} \betainvmod[\mu_2]{b} \betainvmod[\mu_3]{c} \betainvmod[\mu_4]{d} \quad\text{and}\quad \varepsilon^{abcd} = \varepsilon^{\mu_1 \mu_2 \mu_3 \mu_4} \betamod[a]{\mu_1} \betamod[b]{\mu_2} \betamod[c]{\mu_3} \betamod[d]{\mu_4}.
	\] 
	The trace $\Tr$ used to compute the Pontrjagin number associated to the $\g$-valued curvature $\bOm$ is given by
	\[
	\Tr(\varphi\varphi') = - \eta_{ac} \eta_{bd}\varphi_{\mathfrak{so}}^{ab} \varphi_{\mathfrak{so}}^{'cd} + 2 \varphi_{\mathbb{R}} \varphi'_{\mathbb{R}} + \dfrac{2}{\ell\ell'} \eta_{ab} (\varphi_{-1}^a \varphi_{1}^{'b} + \varphi_{1}^a \varphi_{-1}^{'b}) = K_\g(\varphi,\varphi')/4
	\]
in view of the decomposition \eqref{eq:graded-decomp} and \eqref{eq:Killing_Mob_special} below.

    \bigskip
	Taking a linear combination of the type \eqref{eq:action_start} made of the invariant polynomial in $F$ (the Pfaffian for the $\mathfrak{so}$-part) and the Pontrjagin number $\mathdutchcal{P}(\bOm)$ yields the action:
    \begin{align}
		\label{gen Mobius action}
		S_G[\varpi] &=  \int_{\mathcal{M}_i} \bigl(e \Pf (\dfrac{F}{2\pi}) + r \overbrace{\det ( \mathds{1} + \dfrac{\bOm}{2\pi})}^{P(\bOm)} + y \det ( \mathds{1} + \dfrac{\bOmh}{2\pi})\bigr) 
	\end{align} 
        where $\mathcal{M}_i$ is either the $8$-dimensional manifold $\widetilde{\mathcal{M}}$ defined by the Cartan geometry $G/H$ or a $4$-dimensional submanifold $\mathcal{M} \subset\widetilde{\mathcal{M}}$ defined by 
        the distribution:
			\begin{prop}
				\label{submanifold} The subset of vector fields
				$\mathcal{D}_k = \{ X \in \GamTtM | \alpha (X) = \dfrac{k}{2} \beta (X) 
                \}$ defines an integrable submanifold of dimension four $\mathcal{M} \subset\widetilde{\mathcal{M}}$.
		\end{prop}
		
		\begin{proof} Indeed, 
				$\mathcal{D}_k = \{ X \in \GamTtM | \alpha (X) = \dfrac{k}{2} \beta (X)
                \}$  provides us with the vielbein $\alpha - \dfrac{k}{2} \beta$ which corresponds to the four Pfaff forms defining the submanifold $\cm$ such that $\alpha - \dfrac{k}{2} \beta = 0$ on $\cm$.
                For all $X,Y \in \mathcal{D}_k$ one has $ (\Pi - \lambda \wedge \alpha) (X,Y) =  \dfrac{k}{2} ( \Theta + \lambda \wedge \beta ) (X,Y) = \dfrac{k}{2} T(X,Y) $ which yields
				\begin{align*}
					d(\alpha - \dfrac{k}{2} \beta) = - A \wedge (\alpha - \dfrac{k}{2} \beta).
				\end{align*}
				Hence, by the Frobenius theorem \cite[Prop.5.3, p.81]{sharpe2000differential} the corresponding distribution defines $\cm$ as an integrable submanifold of $\widetilde{\mathcal{M}}$. Therefore, there is a subclass of conformal Cartan connections characterizing $\cm$; this subclass can be reached through this kind of "gauge fixing".\footnote{\label{footnote:alpha-beta} We are very indebted to Thierry Masson for his key input in the construction of this integrable distribution defined by this specific Cartan connection. This ought to resemble a gauge fixing which amounts to freezing the degrees of freedom of the $\g_1$-part to be those of the $\g_{-1}$-sector. It can be linked with \cite[Proof of Proposition 3.1, p.285]{sharpe2000differential} where the two tetrads are taken to be linearly dependent each other. To some extent, this selects a mutated Weyl geometry within the mutated Möbius one. One thus gets $\dim\mathcal{M}=11-7=4$.
        One may also consider manifolds of other dimensions by imposing other constraints that reduce the number of degrees of freedom of the soldering form $\varpim$.  } 
		\end{proof}
       \begin{remark}
        The Pfaffian of $F$, that is $\Pf(\dfrac{F}{2\pi})$ is null if $dim(\mathcal{M}_i) \neq 4$.
    \end{remark} 
    For the time being we focus on the case of an integral on $\mathcal{M}_i = \mathcal{M}$ corresponding to the action:
	\begin{align*}
		\label{Mobius action}
		S_G[\varpi] &= r \mathdutchcal{P}(\overline{\Om}) + \int_\mathcal{M} (e \Pf (\dfrac{F}{2\pi}) + y \det ( \mathds{1} + \dfrac{\bOmh}{2\pi}))\\
		&= \int_{\mathcal{M}}  \Bigl( e \Pf( \dfrac{F}{2\pi} ) - \dfrac{r}{8\pi^2} \Tr( \bOm\wedge \bOm ) - \dfrac{y}{8\pi^2} \Tr (\bOmh \wedge \bOmh) \Bigr) \numberthis \\
		%%%%%%%%%%%%%
		& = \int_{\mathcal{M}} \Bigl(
		\dfrac{e }{2 (4\pi)^2} \varepsilon_{abcd} F^{ab} \wedge F^{cd}  
		- \dfrac{r}{8\pi^2} \Tr ( \bOm \wedge \bOm ) 
		- \dfrac{y}{8\pi^2} \Tr ( \bOmh \wedge \bOmh ) 
		\Bigr) \\
		%%%%%%%%%%%%%%%%
		&= \int_{\mathcal{M}} \biggl(
		\dfrac{e }{2(4\pi)^2} \varepsilon_{abcd} \bigl(  R^{ab} \wedge R^{cd} + \dfrac{2}{\ell \ell'} R^{ab} \wedge \phi^{cd} + \dfrac{1}{(\ell \ell')^2} \phi^{ab} \wedge \phi^{cd}
		\bigr) \\
		& \qquad \quad + \dfrac{1}{8 \pi^2} \Bigl( (r+y) \bigl(R^{ab} \wedge R_{ab} + \dfrac{2}{\ell \ell'} R^{ab} \wedge \phi_{ab} + \dfrac{1}{(\ell \ell')^2} \underbrace{\phi^{ab} \wedge \phi_{ab}}_{=0}\bigr)\\
        & \qquad \quad - \dfrac{4r}{\ell \ell'} \eta_{ab} \Pi^a \wedge \Theta^b - 2(r+y)f \wedge f
		\Bigr)
		\biggr) \\
		%%%%%%%%%%%%%%%%
		&= \int_{\mathcal{M}} \biggl(
		\dfrac{e }{2(4\pi)^2} \bigl( \varepsilon_{abcd} R^{ab} \wedge R^{cd} 
		+ \dfrac{4}{\ell\ell'} \sqrt{|\tilde{g}|} \varepsilon_{abcd} \varepsilon^{r_1s_1cs_2} \Rmod[\mu\nu]{ab} \betainvmod[\mu]{r_1} \betainvmod[\nu]{s_1} \alphamod[d]{\mu_4} \betamod[\mu_4]{s_2} d^4x \\
		& \qquad \qquad \qquad \quad + \dfrac{16}{\ell^2\ell'^2} \sqrt{|\tilde{g}|} \varepsilon_{abcd} \varepsilon^{as_1cs_2} \alphamod[b]{\mu_2} \betainvmod[\mu_2]{s_1} \alphamod[d]{\mu_4} \betainvmod[\mu_4]{s_2} d^4x
		\bigr) \\
		& \qquad + \dfrac{1}{8 \pi^2} \Bigl( (r+y) \bigl(R^{ab} \wedge R_{ab} + \dfrac{4}{\ell\ell'} \sqrt{|\tilde{g}|} \varepsilon^{r_1s_1r_2s_2} \Rmod[r_1s_1]{ab} \betamod[c]{\mu_1} \alphamod[d]{\mu_2} \betainvmod[\mu_1]{r_2} \betainvmod[\mu_2]{s_2} \eta_{ca} \eta_{db} d^4x\bigr) \\
		& \qquad \qquad \quad \quad - \dfrac{4r}{\ell \ell'} \eta_{ab} \Pi^a \wedge \Theta^b 
		- 2 (r+y) f \wedge f
		\Bigr)
		\biggr).
	\end{align*} 
    One can observe that unless $\alpha \propto \beta$ we do not retrieve the Palatini term nor the cosmological constant term. 
    
    From the very construction of the action involving polynomials of degree 2 in the curvature $\bOm$ the gauge action is invariant under the duality symmetry $\bOm \leftrightarrow -\bOm$
    \footnote{Considering instead the action on $\mathcal{M}_i = \widetilde{\mathcal{M}}$ one finds the duality symmetry $\bOm \leftrightarrow \sqrt[4]{1}\bOm$.} which can be split as:
    \begin{equation}
		\begin{cases}
			\label{mobius duality}
			F \leftrightarrow - F  \\
			f \leftrightarrow - f  \\
            \{\dfrac{\Pi}{\ell'} \leftrightarrow - \dfrac{\Theta}{\ell}\} \cup \left\{ \{ \dfrac{\Pi}{\ell'} \leftrightarrow - \dfrac{\Pi}{\ell'} \} \cap \{ \dfrac{\Theta}{\ell} \leftrightarrow - \dfrac{\Theta}{\ell} \} \right\}
		\end{cases} 
	\end{equation}
    the action being invariant under each of these transformations.
    Particular transformations satisfying \ref{mobius duality} are:
    \begin{equation}
		\begin{cases}
			\label{mobius duality2}
			R \leftrightarrow - \dfrac{1}{\ell\ell'} (\beta \wedge \bar{\alpha} + \alpha \wedge \bar{\beta})  \\
			  d\lambda \leftrightarrow - \dfrac{1}{\ell\ell'} \bar{\alpha} \wedge \beta \\
            \{ \dfrac{1}{\ell} \beta \leftrightarrow - \dfrac{1}{\ell'} \alpha \}\footnotemark \cup \left\{  \{ \dfrac{\alpha}{\ell'} \leftrightarrow - \dfrac{\alpha}{\ell'} \} \cap \{ \dfrac{\beta}{\ell} \leftrightarrow - \dfrac{\beta}{\ell} \} \right\}
		\end{cases} 
	\end{equation}
    
    \footnotetext{Left hand side of this condition constrains $k=\pm2$ in the context of the identification $\alpha = \dfrac{k}{2} \beta$.}
	
    Thus, upon identifying $\alpha = \dfrac{k}{2} \beta$, the curvature entries reduce to
    \[
    \Theta = T - \lambda \wedge \beta, \quad
    \Pi 
    = \dfrac{k}{2} (T + \lambda \wedge \beta),\quad \phi = k\beta\wedge\bar{\beta} = k\ell^2 \xi,\quad F= R + \dfrac{k\ell}{\ell'} \xi \quad \text{and}\quad f=d\lambda
    \]
    and the action \eqref{Mobius action} reads: 
	\begin{align*}
		\label{gauge action Möbius}
		S_G[\varpi] 
		%%%%%%%%%%%%%%%%%%%%%%%%%
		&= \int_{\mathcal{M}} \biggl(
		\dfrac{e }{2(4\pi)^2} \varepsilon_{abcd} \bigl(  R^{ab} \wedge R^{cd} + \dfrac{2k\ell}{\ell'} R^{ab} \wedge \xi^{cd} +\dfrac{k^2\ell^2}{\ell'^2} \xi^{ab} \wedge \xi^{cd}
		\bigr) \\
		& + \dfrac{1}{8 \pi^2} \Bigl( (r+y) \bigl(R^{ab} \wedge R_{ab} + \dfrac{2k\ell}{\ell'} R^{ab} \wedge \xi_{ab} + \dfrac{k^2\ell^2}{\ell'^2} \underbrace{\xi^{ab} \wedge \xi_{ab}}_{=0}\bigr) - \dfrac{2rk}{\ell \ell'} \eta_{ab} \Theta^a \wedge \Theta^b - 2(r+y)f^2
		\Bigr)
		\biggr) \\
		%%%%%%%%%%%%%%%%%%%%%%%%%%
		& = \int_{\mathcal{M}} \biggl(
		\dfrac{e }{2(4\pi)^2} \varepsilon_{abcd} \bigl(  R^{ab} \wedge R^{cd} 
		+ \dfrac{2k\ell}{\ell'} R^{ab} \wedge \xi^{cd} +\dfrac{k^2\ell^2}{\ell'^2} \xi^{ab} \wedge \xi^{cd}
		\bigr)  \\
		& + \dfrac{1}{8 \pi^2} \Bigl( (r+y) \bigl(R^{ab} \wedge R_{ab} 
		+ \dfrac{2k\ell}{\ell'} R^{ab} \wedge \xi_{ab}\bigr) 
		- \dfrac{2rk}{\ell \ell'} \eta_{ab} T^a \wedge T^b 
		% %
		% %
		- 2 (r+y) d\lambda \wedge d\lambda
		\Bigr)
		\biggr) \\
		%%%%%%%%%%%%%%%%%%%%%%%%%%%
		&= \int_{\mathcal{M}} \biggl(
		\dfrac{e }{2(4\pi)^2} \bigl( \varepsilon_{abcd} R^{ab} \wedge R^{cd} 
		+ \dfrac{8k}{\ell\ell'} \sqrt{|\tilde{g}|} \Rmod[\mu\nu]{ab} \betainvmod[\mu]{a} \betainvmod[\nu]{b} d^4x 
		+  \dfrac{24k^2}{\ell^2\ell'^2} \sqrt{|\tilde{g}|} d^4x
		\bigr) \\
		& \qquad \quad + \dfrac{1}{8 \pi^2} \Bigl( (r+y) \bigl(R^{ab} \wedge R_{ab} 
		+ \dfrac{2k}{\ell\ell'} \sqrt{|\tilde{g}|} \varepsilon^{rs}_{~~ab} \Rmod[\mu\nu]{ab} \betainvmod[\mu]{r} \betainvmod[\nu]{s} d^4x\bigr) \\
		& \qquad \qquad \quad \quad - \dfrac{2rk}{\ell \ell'} \eta_{ab} T^a \wedge T^b 
		% %
		% %
		- 2 (r+y) d\lambda \wedge d\lambda
		\Bigr)
		\biggr). \numberthis
	\end{align*}
    
	The resulting action for $k=0$ reads:
	\begin{align}
		S_G[\varpi] &= \int_{\mathcal{M}} \Bigl(
		\dfrac{e }{2(4\pi)^2} \varepsilon_{abcd} R^{ab} \wedge R^{cd}
		+ \dfrac{1}{8 \pi^2} \bigl( (r+y) R^{ab} \wedge R_{ab} 
		- 2 (r+y) d\lambda \wedge d\lambda
		\bigr)
		\Bigr) .
	\end{align}
	and corresponds to the action \eqref{eq:action-ISO} for $G=ISO(3,1) \Leftrightarrow k = 0$, up to the $\lambda$ kinetic term. However, ignoring the boundary terms leaves the equations of motion unchanged if we keep the same matter action as in the preceding example \ref{Matter action}.

    \smallskip
	While for $k \neq 0$, a direct comparison between $S'$ \eqref{action Perez} and $S_G$ \eqref{gauge action Möbius} amounts to identifying the 6 coupling constants present in $S'$ via the system:
	\begin{equation}
		\begin{cases}
			\label{system Möbius}
			\dfrac{e k}{8\pi^2 \ell\ell'}  = \alpha_1
			\\
			\dfrac{(r+y) k }{4\pi^2 \ell\ell'} = (\alpha_2 - \alpha_5)
			\\
			\dfrac{r +y}{8\pi^2} = \alpha_3 
			\\
			\dfrac{e }{(4\pi)^2} = \alpha_4
			\\
			- \dfrac{r k }{4\pi^2 \ell\ell'} = \alpha_5  
			\\
			\dfrac{e k^2}{32\pi^2(\ell\ell')^2}  = \alpha_6 
		\end{cases}
		\Leftrightarrow
		\begin{cases}
			\alpha_1 = \dfrac{2k}{\ell\ell'} \alpha_4 = \dfrac{4\ell\ell'}{k} \alpha_6  \\
			\alpha_2 - \alpha_5 = \dfrac{2k}{\ell\ell'} \alpha_3 = - (1+ \dfrac{y}{r}) \alpha_5 \\
			\alpha_1 = \dfrac{e k}{(r+y)\ell\ell'} \alpha_3 \\
			r +y= \dfrac{4\pi^2\ell\ell'}{k} (\alpha_2 -\alpha_5)
		\end{cases}
	\end{equation}
	The solution of \eqref{system Möbius} in terms of the parameters $e$, $r$, $k$, $\ell$ and $\ell'$ is: 
	\begin{equation}
		\begin{cases}
			\label{solution Möbius}
			\alpha_1 = \dfrac{2k}{\ell\ell'} \alpha_4 = \dfrac{4\ell\ell'}{k} \alpha_6  \\
			\alpha_2 = \dfrac{ky}{4\pi^2 \ell\ell'} =\dfrac{2k y}{(r+y)\ell\ell'} \alpha_3 = - \dfrac{y}{r} \alpha_5 \\
			\alpha_1 = \dfrac{e}{2y} \alpha_2 .
		\end{cases}
	\end{equation}
	Hence, the resulting action reads:
	\begin{align*} \label{eq:action-Moebius-bis}
		S_G[\varpi] 
		= \int_{\mathcal{M}} &\Bigl(
		\overbrace{ \dfrac{k }{4\pi^2\ell\ell'} \big( \dfrac{e}{4} \underbrace{R^{ab} \wedge \beta^c\wedge\beta^d \varepsilon_{abcd}}_{Palatini} + y R^{ab} \wedge \beta_a\wedge\beta_b \big)}^{Holst}
		+ \overbrace{\dfrac{e k^2}{32\pi^2(\ell\ell')^2}  \beta^a\wedge\beta^b \wedge \beta^c\wedge\beta^d \varepsilon_{abcd}}^{Bare~Cosmological~constant} \nonumber
		\\
		%%%%%%%%
		& \qquad + \overbrace{\dfrac{r +y}{8 \pi^2} R^{ab} \wedge R_{ab}}^{Pontrjagin}
		+ \overbrace{\dfrac{e }{2(4\pi)^2} R^{ab} \wedge R^{cd} \varepsilon_{abcd}}^{Euler} 
		\overbrace{- \dfrac{r k }{4 \pi^2 \ell\ell'} (T^a \wedge T_a - R^{ab} \wedge \beta_a\wedge\beta_b)}^{Nieh-Yan}
		\Bigr) \nonumber \\
		% %
		% %
		&\qquad \overbrace{- \dfrac{r+y}{4\pi^2} d\lambda \wedge d\lambda}^{Dilation ~ kinetic ~ term}
		\Bigr) . \numberthis
	\end{align*}

    This action corresponds to \eqref{topological action} with the substitution $k/\ell^2 \to k/(\ell\ell')$ plus an additional kinetic term for the dilations. We will show in the next two subsections that in the presence of matter:
    \begin{itemize}
        \item Action \eqref{eq:action-Moebius-bis} gives back the equations of motion \eqref{eom frame Lorentz}-\eqref{eom spin Lorentz}.
        \item Adding the term $-\, \dfrac{v k}{2\pi^2\ell \ell'} \eta_{ab} T^a \wedge \lambda \wedge \beta^b$ 
     in the action \eqref{eq:action-Moebius-bis} leads to a second source of curvature expressed in terms of the spin density of matter, torsion and their variations for the Einstein equations.
    \end{itemize}

  \subsubsection{Fermionic matter field action}
    
    Choosing the symmetric decomposition $ \g= \g_0 \oplus \m $ with $\m = \g_{1} \oplus \g_{-1}$ and $\g_0 = \h = \mathfrak{co}(3,1)$ we compute the Killing metric:
		\begin{align} \label{eq:Killing_Mob_special}
		   K_\g (\varphi,\varphi') = \overbrace{- 4 \eta_{ac} \eta_{bd}\varphi_{\mathfrak{so}}^{ab} \varphi'{}_{\!\!\mathfrak{so}}^{\,cd} + 8 \varphi_{\mathbb{R}} \varphi'_{\mathbb{R}}}^{K_0} + \overbrace{ \dfrac{8}{\ell\ell'} \eta_{ab} ( \varphi_{-1}^a \varphi'{}_{\!\!1}^{\,b} + \varphi_{1}^a \varphi'{}_{\!\!-1}^{\,b} )}^{K_\m} 
		\end{align}  
	   On the 4-dimensional submanifold $\mathcal{M}$ defined by the gauge fixing $\alpha = \dfrac{k}{2} \beta$, one has $\varpim = \beta^a M_a$ and one gets for the metrics
	\begin{align}
        \text{on}~\g: & & & h(\varphi,\varphi') = \kappa K_\g (\varphi,\varphi')\\
        \text{on}~\mathcal{M}: & & & g (X,Y)  = \zeta (\varpim^*h) (X,Y) = \kappa \zeta (\varpim^*K_\m) (X,Y) = \dfrac{8 k \kappa \zeta}{\ell\ell'}\, \eta \bigl( \beta(X) , \beta(Y) \bigr) .
	\end{align}

	If one considers the usual covariant derivative $D_{\mathfrak{so}}=(\partial_\mu + \dfrac{1}{4} A_{ab,\mu} \gamma^a \gamma^b) dx^\mu$, the matter action similar to \eqref{Matter action} is then given by \footnote{Where we set $\kappa^2 \zeta = \tfrac{(\ell\ell')^2}{384k^2}$ such that the overall factor is arranged to stick with the Dirac action given in \cite[p.197]{gockeler_differential_1987}.}:
	\begin{align}
        \hskip -8mm
		S_M[A,\beta] & = 
        \int_{\mathcal{M}} h^\varepsilon\bigl( i \psi, h_{ab} \gamma^a \varpim^b \wedge *D_\mathfrak{so}\psi\bigl)
        = \int_{\mathcal{M}} h^\varepsilon\bigl( i \psi, \eta_{ab} \dfrac{64k^2\kappa^2\zeta}{(\ell\ell')^2} \gamma^a \beta^b \wedge \tilde{*}D_\mathfrak{so}\psi\bigl)
        = \mathrm{Re} S_D\,.
	\end{align}
    Let $C_i$ be the $i$-dimensional Clifford algebra. We can extend the isomorphism \ref{isomorphism} between $\mathfrak{so}(n-1,1)$ and $\mathfrak{spin}(n-1,1)$ to the morphism $\Phi_\varepsilon : \mathfrak{gl}(n,\mathbb{R}) \rightarrow C_2 \oplus C_0$. Let $M = M^{rs} e_r \bar{e}_s \in \mathfrak{gl(n,\mathbb{R)}}$. We define:
    \begin{align}
    \label{isomorphism gen}
		\Phi_\varepsilon (M) 
        &= M^{rs} \Phi_\varepsilon (e_r \bar{e}_s)
        = \dfrac{M^{rs}}{4} \bar{e}_a e_r \bar{e}_s e_b \gamma^a \gamma^b
        = \dfrac{M^{rs}}{4} \eta_{ar} \eta_{sb} \gamma^a \gamma^b \\
        %%%%%%%%%%%%%%%%%%%%%%
        &= \dfrac{1}{4} \bigl(\sum_{a<b} M_{ab} \gamma^a \gamma^b + \sum_{a>b} M_{ab} \gamma^a \gamma^b + \sum_{a} M_{aa} (\gamma^a)^2 \bigr) \\*
        &=  \sum_{a<b} \dfrac{1}{4} (M_{ab} - M_{ba}) \gamma^a \gamma^b +  \dfrac{M_{ab}}{4} \eta^{ab} 
	\end{align}
where we used the relation 
    $(\gamma^a)^2 = 
    \left\{
		\begin{array}{lcll}
			+1 \text{ if } a \leq n-1 \\
			-1 \text{ if } a>n-1
		\end{array}
	\right.$.
    A particular case is the isomorphism between elements of the algebra $\g_0 = \mathfrak{so}(n-1,1) \oplus \mathbb{R}$
    and those of $\mathfrak{spin}(n-1,1) \oplus \mathbb{R}$ which corresponds to $\Phi_\varepsilon : \mathfrak{so}(n-1,1) \oplus \mathbb{R} \rightarrow \mathfrak{spin}(n-1,1) \oplus \mathbb{R}$. The action of  $\Phi_\varepsilon$ on the generators of $\mathfrak{so}(n-1,1)$ and on the identity matrix which generates the dilation part of $\g$ is given by:
    \begin{align}
		\Phi_\varepsilon (J_{rs} \oplus \mathds{1}_n) = \dfrac{1}{4} \bar{e}_a (J_{rs} \oplus \mathds{1}) e_b \gamma^a \gamma^b = \dfrac{1}{4} \bigl( (J_{rs})_{ab} \oplus \bar{e}_a e_r \bar{e}_s e_b \delta^{rs} \bigr) \gamma^a \gamma^b =  \dfrac{1}{2} \eta_{ar} \eta_{bs} \gamma^a \gamma^b \oplus \dfrac{n-2}{4} .
	\end{align}
    
	Thus, a more natural matter action $\tS_M$ with covariant derivative $D_0 = d + \Phi_\varepsilon (\varpiz) = d + \Phi_\varepsilon (A - \lambda \mathds{1}_4)
	= (\partial_\mu + \dfrac{1}{4} A_{ab,\mu} \gamma^a \gamma^b - \dfrac{1}{2} \lambda_\mu) dx^\mu$ taking into account the full action of the
    $ \g_0 = \mathfrak{co}(3,2) $-valued part of the connection $\varpiz$ is:
	
	\begin{align}
		\tS_M[A,\beta,\lambda] & =
        \int_{\mathcal{M}} h^\varepsilon\bigl( i \psi, h_{ab} \gamma^a \varpim^b \wedge *D_0\psi\bigl) = \overbrace{S_M}^{\mathrm{Re} S_D} - 
        \dfrac{1}{2}\int_{\mathcal{M}} h^\varepsilon\bigl( i \psi, h_{ab} \gamma^a \varpim^b \wedge *(\lambda\psi)\bigl).
	\end{align}

   \subsubsection{The total action}
    
	The study of the full action $\tS_T= S_G + \tS_M$ is left for future work. In the following we will restrict ourselves to the study of the action $S_T = S'_G + S_M = S_G + S_M - \int_\mathcal{M} \dfrac{v k}{2\pi^2\ell \ell'} \eta_{ab} T^a \wedge \lambda \wedge \beta^b$.
	
	For $k \neq 0$ the equations of motion for $\beta$, $A$ and $\lambda$ yield (neglecting boundary terms):
	
	\begin{align*}
		\dfrac{\delta \mathcal{L}'_G [A,\beta,\lambda]}{\delta \beta^c}
		%%%%%%%
		& = \dfrac{ek}{8\pi^2\ell\ell'} ( R^{ab} + \dfrac{ k }{\ell\ell'} \beta^a \wedge \beta^b ) \wedge \beta^d \varepsilon_{abcd} 
		+ \dfrac{k}{2\pi^2\ell\ell'} ( y R^{ab} \wedge \beta_b \eta_{ac} + 2v T_c \wedge \lambda - v \beta_c \wedge d\lambda) \\
		%%%%%%%%
		&= \tau_c = - \dfrac{\delta \mathcal{L}_M [A,\beta,\lambda]}{\delta \beta^c} \label{eom frame Möbius} \numberthis \\
		%%%%%%%%%%%%%%%%%%%%%%%%%%%%%%%%%%%%%%%%%%%%%%%%%%%%%%%%%
		\dfrac{\delta \mathcal{L}'_G [A,\beta,\lambda]}{\delta A^{ab}}
		%%%%%%%%
		& = 
		\dfrac{ek}{8\pi^2 \ell\ell'}T^c \wedge \beta^d \varepsilon_{abcd} 
		+ \dfrac{k}{2\pi^2\ell\ell'} \bigl( y T_a - v \lambda \wedge \beta_a \bigr) \wedge \beta_b
		%%%%%%%%%
		= \dfrac{1}{2} \mathfrak{s}_{ab} = - \dfrac{\delta \mathcal{L}_M [A,\beta,\lambda]}{\delta A^{ab}} \label{eom spin Möbius} \numberthis \\
		%%%%%%%%%%%%%%%%%%%%%%%
		\dfrac{\delta \mathcal{L}'_G[A,\beta, \lambda ]}{\delta \lambda} & = - \dfrac{vk}{2\pi^2\ell\ell'} T^a \wedge \beta_a = 0 = - \dfrac{\delta \mathcal{L}_M [A,\beta,\lambda]}{\delta \lambda}\ .
		\numberthis \label{eom scalar}
	\end{align*}
	Equation \eqref{eom scalar} implies either $v=0$ or $T^a \wedge \beta_a=0$. 

    \smallskip
	For $v=0$, one then recovers the usual equations of motion associated to Einstein-Cartan gravity with a Holst action and bare cosmological constant term, and torsion $T$ linked to the spin density of matter.
	
	In the other case, using the Hodge star operator lets us show that $T^a \wedge \beta_a=0$ is equivalent to requiring the axial vector part of torsion to vanish ($\tilde{t}^a = \dfrac{1}{6} T^b_{~\mu \nu} \varepsilon^{a~\mu \nu}_{~b}=0$) as is described in \cite[Table 3]{gronwald_gauge_1996}. The end result being that the new torsion $T^a_{~bc} = Z^a_{~bc} + \delta^a_{[b} t^{\vphantom{a}}_{c]}$ can be expressed solely in terms of $20$ independent components instead of the usual $24$. Here $Z^a_{~bc}$ is a $3$-indices tensor and $t_c$ is a $4$-dimensional vector with respectively $16$ and $4$ independent components.

 \smallskip
    A very particular case satisfying $T^a \wedge \beta_a=0$ is $T=0$, the equations of motion then become:	
	\begin{align}
		\label{EOM Möbius 2}
		\dfrac{ek}{8\pi^2\ell\ell'} ( R^{ab} + \dfrac{ k }{\ell\ell'} \beta^a \wedge \beta^b ) \wedge \beta^d \varepsilon_{abcd} 
		+ \dfrac{k}{2\pi^2\ell\ell'} \bigl( y R^{ab} \wedge \beta_b \eta_{ac} - v \beta_c \wedge d\lambda \bigr) & = \tau_c \\
		%%%%%%%%%%%%%%%%%%%%%%%%%%%%%%%%%%%%%%%%%%%%%%%%%%%%%%%%%
		-\dfrac{vk}{\pi^2\ell\ell'} \lambda \wedge \beta_a \wedge \beta_b
		%%%%%%%%%
		&= \mathfrak{s}_{ab} \label{eom spin 2 Möbius} \\
		%%%%%%%%%%%%%%%%%%%%%%%
		T &= 0 
	\end{align}
	where equation \eqref{eom spin 2 Möbius} links spin density to the scalar gauge fields for dilation: 
	\begin{align}
		\label{link spin density scalar}
		\lambda_c = \lambda_\mu \betainvmod[\mu]{c} =- \dfrac{\pi^2 \ell \ell'}{ 6 vk } \mathfrak{s}^{ab}_{~~abc} \ .
	\end{align}
	
	Using the relation $\gamma = \dfrac{2\alpha_1}{2\alpha_2 -\alpha_5}$, given in \cite[p13]{rezende_4d_2009}, we identify $\gamma = \dfrac{e}{r+2y}$, with $\gamma$ the Barbero-Immirzi parameter.
	
	Provided we set $ e = \dfrac{8 \pi^2 \alpha_1 \ell\ell'}{k} = \dfrac{3\pi}{2\Lambda_0 G} $ and $ \ell \ell' = - \dfrac{3k }{ \Lambda_0 } $ where $\alpha_1 = - \dfrac{1}{16 \pi G}$, applying the Hodge star $*$ to \eqref{EOM Möbius 2} and replacing $\lambda$ according to \eqref{link spin density scalar} gives us the Einstein equations modified by the Barbero-Immirzi part of the action with bare cosmological constant $\Lambda_0$ and with matter sources given by the energy-momentum tensor and a particular combination of derivatives of the spin density of matter (the latter resulting from the presence of the dilation 1-form $\lambda$ of the connection):	
	\begin{align}
		& G_{kc} 
		+ 
        (\dfrac{1}{2\gamma} - \dfrac{\Lambda_0 Gr}{3\pi}) R^{ab}_{~~rs} \varepsilon^{rs}_{~~kb} \eta_{ac}
		+ \Lambda_0 \eta_{kc}
		%%%%
		= - \dfrac{1}{2 \alpha_1 } ( \tau_{ck} - \dfrac{1}{24} \partial_{\mu} \mathfrak{s}^{ab}_{~~ab\nu} \varepsilon^{\mu\nu}_{~~\;ck} ) .
		\label{eom frame 3}
	\end{align}
 
	\begin{remark}
		In this particular case the fact that Torsion is null implies that the Einstein tensor $G_{kc}$ is symmetric (see \cite[p3]{schucker_torsion_2012}  for more details).
		When taking into account that $\eta_{kc}$ is symmetric as well, the anti-symmetricity of $ \partial_{\mu} \mathfrak{s}^{ab}_{~~ab\nu} \varepsilon^{\mu\nu}_{~~\;ck} $ implies that the energy-momentum tensor $\tau_{ck}$ sports an anti-symmetric part if $ \partial_{\mu} \mathfrak{s}^{ab}_{~~ab\nu} \varepsilon^{\mu\nu}_{~~\;ck} \neq 0$. This confirms that we can consider a matter Lagrangian that depends on the spin connection $A$ as we explicitly did.
	\end{remark}

	Separating the energy-momentum tensor $\tau_{ck}$ as $\tau_{ck}=\tau_{M,ck} - \rho_{\text{vac}} \eta_{ck}$ with $-\rho_{\text{vac}}\eta_{ck}$ the part related to the vacuum energy density $\rho_{\text{vac}}\geq0$ then leads to:
	\begin{align}
		& G_{kc} 
		+ (\dfrac{1}{2\gamma} - \dfrac{\Lambda_0 Gr}{3\pi}) R^{ab}_{~~rs} \varepsilon^{rs}_{~~kb} \eta_{ac}
		+ \Lambda \eta_{kc}
		%%%%
		= - \dfrac{1}{2 \alpha_1 } ( \tau_{M,ck} - \dfrac{1}{24} \partial_{\mu} \mathfrak{s}^{ab}_{~~ab\nu} \varepsilon^{\mu\nu}_{~~\;ck} ) \label{eom frame 4}
	\end{align}
	with $\Lambda = \Lambda_0 + \Lambda_{\text{vac}} = \overbrace{ - \dfrac{3k}{\ell\ell'}}^{\Lambda_0} \overbrace{-\dfrac{\rho_{\text{vac}}}{2\alpha_1}}^{\Lambda_{\text{vac}}} = - \dfrac{3k}{\ell\ell'} - \dfrac{4\pi^2\ell\ell'\rho_{\text{vac}}}{ek} = \Lambda_0 + 8 \pi G \rho_{\text{vac}} $ the effective cosmological constant, and $\Lambda_{\text{vac}}$ the vacuum contribution to the cosmological constant.
	
	One can infer that changes in cosmological history of the rate of expansion of the universe could be partly linked to spin-density variations in this way.

    If one does not assume $T=0$ the EOM's instead yield:
    \begin{align}
		G_{kc} 
		  &+ (\dfrac{1}{2\gamma}- \dfrac{\Lambda_0 Gr}{3\pi}) R^{ab}_{~~rs} \varepsilon^{rs}_{~~kb} \eta_{ac}
		+ \Lambda_0 \eta_{kc} \notag \\
		%%%%
		&= - \dfrac{1}{2 \alpha_1 } \Bigl( \tau_{ck} - \dfrac{1}{24} (\eta_{cs}  \beta^s_\nu \partial_{\mu} - T_{c,\mu\nu} ) \bigl( \mathfrak{s}^{ab}_{~~abr} +
        ( \dfrac{1}{2\pi G\gamma} -\dfrac{\Lambda_0r}{3\pi^2}  ) T^m_{~rm} \bigr)\varepsilon^{\mu\nu r}_{~~~\,k} \Bigr) \label{eq:pourRmk} \\[2mm]
		%%%%%%%%%%%%%%%%%%%%%%%%%%%%%%%%%%%%%%%%%%%%%%%%%%%%%%%%%
		\lambda_r &= 
        \dfrac{1}{v} \bigl( \dfrac{\pi}{16 \Lambda_0 G_0} \overbrace{T^c_{~ab} \varepsilon^{ab}_{~~cr}}^{-6\tilde{t}_r = 0} +  \dfrac{y}{3}  T^c_{~rc} + \dfrac{\pi^2}{2\Lambda_0} \mathfrak{s}^{ab}_{~~abr} \bigr) 
        = \dfrac{1}{v} \bigl( \dfrac{y}{3}  T^c_{~rc} + \dfrac{\pi^2}{2\Lambda_0} \mathfrak{s}^{ab}_{~~abr} \bigr) \label{eom spin 3 Möbius} \\
		%%%%%%%%%%%%%%%%%%%%%%%
		\tilde{t^a} &= 0 .
	\end{align}
	
    \begin{remark}
	    The secondary source of curvature $\dfrac{1}{48 \alpha_1 } (\eta_{cs}  \beta^s_\nu \partial_{\mu} - T_{c,\mu\nu} ) \bigl( \mathfrak{s}^{ab}_{~~abr} + ( \dfrac{1}{2\pi G\gamma} -\dfrac{\Lambda_0r}{3\pi^2}  ) T^m_{~rm} \bigr)\varepsilon^{\mu\nu r}_{~~~\,k} $ in \eqref{eq:pourRmk} is highly dependent on the value of $\Lambda_0$ and~$r$.        
        One can for example tune $r=\dfrac{3\pi}{2\Lambda_0G\gamma}$ such that it is now dominated by the term $\dfrac{1}{48 \alpha_1 } (\eta_{cs}  \beta^s_\nu \partial_{\mu} - T_{c,\mu\nu} ) \mathfrak{s}^{ab}_{~~abr} \varepsilon^{\mu\nu r}_{~~~\,k} $.
	\end{remark}

    In case $v=0$ we retrieve the usual equations of motion \eqref{eom lorentz 1}-\eqref{eom lorentz 2}. We will show in the next section the action $S_G + S_M$ is asymptotically topological in the limit $\Lambda_0 \rightarrow 0$.
    
    \bigskip
   
	Replacing $ e = \dfrac{8 \pi^2 \alpha_1 \ell\ell'}{k} = \dfrac{3\pi}{2\Lambda_0 G} $; $ \ell \ell' = - \dfrac{3k }{ \Lambda_0 } $ and $y=\dfrac{3\pi }{ 4\Lambda_0 G \gamma}  - \dfrac{r}{2}$ in the action \eqref{eq:action-Moebius-bis}, we obtain:
	\begin{align*}
		S_G[\varpi] &= \int_\mathcal{M} \bigl( r P(\bOm) + \dfrac{3\pi}{2\Lambda_0 G}  \Pf (\dfrac{F}{2\pi}) + (\dfrac{3\pi}{4\Lambda_0 G\gamma} -\dfrac{r}{2}) \det ( \mathds{1} + \dfrac{\bOmh}{2\pi})
        \bigr) \numberthis \\
		&= \int_{\mathcal{M}} \biggl(\dfrac{\alpha_1}{2} \biggl(
		\overbrace{ \underbrace{R^{ab} \wedge \beta^c\wedge\beta^d \varepsilon_{abcd}}_{Palatini} + (\dfrac{2}{\gamma} - \dfrac{4 \Lambda_0 G r}{3\pi}) R^{ab} \wedge \beta_a\wedge\beta_b }^{Holst}
		\overbrace{-\dfrac{\Lambda_0}{6}  \beta^a \wedge \beta^b \wedge \beta^c \wedge \beta^d \varepsilon_{abcd}}^{Cosmological~constant} 
		\\
		& \qquad \qquad \overbrace{- \dfrac{3}{2\Lambda_0} R^{ab} \wedge R^{cd} \varepsilon_{abcd} }^{Euler}
		\overbrace{ \underbrace{
        - (\dfrac{2Gr}{\pi} + \dfrac{3}{\Lambda_0 \gamma})}_{\tfrac{2\vartheta}{\alpha_1}} R^{ab} \wedge R_{ab} }^{Pontrjagin}
		\overbrace{-\dfrac{8\Lambda_0Gr}{3\pi} \bigl( T^a \wedge T_a - R^{ab} \wedge \beta_a \wedge \beta_b \bigr)}^{Nieh-Yan} \Bigr) \\
		% %
		% %
		&\qquad \qquad \overbrace{-(\dfrac{r}{8\pi^2} + \dfrac{3}{16\pi\Lambda_0G\gamma}) d\lambda \wedge d\lambda}^{\text{ Dilation kinetic term }}
		\biggr)
	\end{align*}
	
	\begin{remark}
		The interpretation with respect to Chern-Simons modified gravity remains, only the potential term changes:
		\begin{align}
			V(\vartheta) &=  
			(\dfrac{1}{\gamma} - \dfrac{ 2\Lambda_0 G r}{3\pi}) \alpha_1 R^{ab} \wedge \beta_a\wedge\beta_b 
			\overbrace{- \dfrac{\Lambda_0\alpha_1}{12} \varepsilon_{abcd} \beta^a \wedge \beta^b \wedge \beta^c \wedge \beta^d }^{Cosmological~constant}
			\overbrace{- \dfrac{3\alpha_1}{4\Lambda_0} R^{ab} \wedge R^{cd} \varepsilon_{abcd} }^{Euler} \nonumber \\
			&\qquad \overbrace{-\dfrac{4\Lambda_0G\alpha_1r}{3\pi} \bigl( T^a \wedge T_a - R^{ab} \wedge \beta_a \wedge \beta_b \bigr)}^{Nieh-Yan}
			% %
			% %
			\overbrace{-(\dfrac{r}{8\pi^2} + \dfrac{3}{16\pi\Lambda_0G\gamma}) d\lambda \wedge d\lambda}^{\text{ Dilation kinetic term }} .
		\end{align}
		
	\end{remark}
	
	Setting $\lambda=0$, we retrieve the action \eqref{topological action mutation}, namely: 
	\begin{align*}
		S_G[\varpi] &= \int_{\mathcal{M}} \dfrac{\alpha_1}{2} \biggl(
		\overbrace{ R^{ab} \wedge \beta^c\wedge\beta^d \varepsilon_{abcd} + (\dfrac{2}{\gamma} - \dfrac{4 \Lambda_0 G r}{3\pi}) R^{ab} \wedge \beta_a\wedge\beta_b }^{Holst}
		\overbrace{-\dfrac{\Lambda_0}{6}  \beta^a \wedge \beta^b \wedge \beta^c \wedge \beta^d \varepsilon_{abcd}}^{Cosmological~constant} 
		\\
		& \qquad \qquad \overbrace{- \dfrac{3}{2\Lambda_0} R^{ab} \wedge R^{cd} \varepsilon_{abcd} }^{Euler}
		\overbrace{
        - (\dfrac{2Gr}{\pi} + \dfrac{3}{\Lambda_0 \gamma}) R^{ab} \wedge R_{ab} }^{Pontrjagin}
		\overbrace{-\dfrac{8\Lambda_0Gr}{3\pi} \bigl( T^a \wedge T_a - R^{ab} \wedge \beta_a \wedge \beta_b \bigr)}^{Nieh-Yan}
		\biggr) .
	\end{align*}
	In particular, this shows that by imposing specific constraints on the components $\alpha$ and $\lambda$ of the Cartan connection $\varpi$ for $G=SO(4,2)/\{\pm I\}$, one can recover the action for the subgroups $SO(4,1),
	~SO(3,2),~\text{or } ISO(3,1)$ already considered in Section~\ref{section SO 5 ETC}. Indeed, one has with the respective constraints:
	\begin{align}
        \{\lambda = 0 \text{ and }
			\alpha = \dfrac{k}{2} \beta\}\, \cap
		\begin{array}{ll}
			%%%%%%%%%%%%%%%%%%%%%
			\left\{ 
			\begin{array}{ll}
				sign(k\ell \ell') = -, \quad \Lambda_0 = - \dfrac{3k}{\ell \ell'} > 0 & \Leftrightarrow G = 	SO(4,1) \\
				k=0  & \Leftrightarrow G = ISO(3,1) \\
				sign(k \ell \ell') = +, \quad \Lambda_0 = - \dfrac{3k}{\ell \ell'} < 0 & \Leftrightarrow G = SO(3,2).
				%%%%%%
			\end{array}		
			\right.
		\end{array}
	\end{align}

    \subsubsection*{$8$-dimensional manifold $\widetilde{\mathcal{M}}$ action}

    In the case $\alpha \neq \dfrac{k}{2} \beta$ corresponding to the $8$-dimensional manifold $\widetilde{\mathcal{M}}$ ($\alpha$ is related to the special conformal transformations) the action \eqref{gen Mobius action} becomes:
	\begin{align}
        \label{gen Möbius action 2}
		S_G[\varpi] &=  \int_{\tilde{\mathcal{M}}} \bigl(\overbrace{r \det ( \mathds{1} + \dfrac{\bOm}{2\pi})}^{P(\bOm)} + y \det ( \mathds{1} + \dfrac{\bOmh}{2\pi})\bigr) \\
		&= \int_{\tilde{\mathcal{M}}}  \Bigl( 
        - \dfrac{r}{128\pi^2} \bigl(\Tr( \bOm^2)^2 - 2 \Tr (\bOm^4) \bigr) 
        - \dfrac{y}{128\pi^2} \bigl(\Tr( \bOmh^2)^2 - 2 \Tr (\bOmh^4) \bigr) \Bigr) \nonumber
	\end{align}

    Study of action \ref{gen Möbius action 2} shall be the subject of another article.

\medskip

    \section{Topological total action}
    \label{Topological total action}
    
    In this section we discuss the topological feature of the actions constructed according to the combination \eqref{eq:action_start} we started with. In doing so, we have considered\footnote{This ought to be considered as a "deformation" of the standard Chern-Simons theory.} the evaluation of some invariant symmetric polynomials $P(\bOmh^n)$ on the $\h$-part of the Cartan curvature $\bOm$ which takes values in a symmetric Lie algebra $\g=\h\oplus\m$. One can establish the following
    \begin{theorem}
    \label{thm cond topo conn}
        Let $H\subset G$ be Lie groups such that $G/H$ defines a symmetric space with symmetric Lie algebra $\g = \h \oplus \m$ and $\mathcal{M}$ be a manifold of dimension $m$. 
    Let $\varpi=\varpih \oplus \varpim$ be a Cartan connection associated to the Cartan geometry modeled on $G/H$ with curvature $\bOm = d\varpi + \tfrac{1}{2} [\varpi,\varpi] = \bOmh \oplus \bOmm$. Let $\varpi_t = t \varpi$ with curvature $\bOm_t = d\varpi_t + \tfrac{1}{2} [\varpi_t,\varpi_t]$.

\smallskip
    Regarding the invariant polynomials $P$ used in our construction two cases are in order:

\medskip\noindent
       Case 1. Let $P\in{\cal I}^n(G)$ be a $G$-invariant polynomial of degree $n$. Restriction of $P$ to $H$ yields $P\in{\cal I}^n(H)$. One has
       \begin{align*}
           (a) &\ dP(\bOmh^n) = -n P([\varpim,\bOmm], \bOmh^{n-1}) = n(n-1) P([\varpim,\bOmh], \bOmh^{n-2},\bOmm) \\
           (b) &\ P(\bOmh^n) + \sum_{k=1}^n \begin{pmatrix} n \\ k \end{pmatrix} P(\bOmh^{n-k},\bOmm^k) = d Q_{2n-1}(\varpi) \ \mathrm{with}\ Q_{2n-1}(\varpi) = n \int_0^1 dt P(\varpi,\bOm^{n-1}_t).
       \end{align*}

\medskip\noindent
       Case 2. Let $P\in{\cal I}^n(H)$ be an $H$-invariant polynomial of degree $n$. This case encapsulates the previous one if the $G$-invariance is not fully taken into account. One has
        \begin{align*}
           (c) &\ dP(\bOmh^n) = -n P([\varpim,\bOmm], \bOmh^{n-1}) \\
           (d) &\ \dfrac{d}{dt} P (\bOm_{t,\h}^n) = n d P(\varpih, \bOm_{t,\h}^{n-1}) + n P([\varpi_{t,\m},\varpim], \bOm_{t,\h}^{n-1}) + n(n-1) P (\varpih, [\bOm_{t,\m},\varpi_{t,\m}],\bOm_{t,\h}^{n-2}).
       \end{align*}
    \end{theorem}
    
    \begin{proof}
        The proof relies on the ones used for the Chern-Weil theorem in \cite{nakahara_geometry_2018} or for the Chern-Simons theory related to our situation \cite{donnelly1977chern,kobayashi1971g}. 
         Let us first compute
       \begin{align} \label{eq:dP}
        dP(\bOmh^n) &= n P(d\bOmh,\bOmh^{n-1}) \notag\\
        &= n P(d\bOmh+[\varpih,\bOmh]+[\varpim,\bOmm],\bOmh^{n-1}) - 
        n P([\varpih,\bOmh],\bOmh^{n-1}) - n P([\varpim,\bOmm],\bOmh^{n-1}) \notag\\
        &= - n P([\varpim,\bOmm],\bOmh^{n-1})
       \end{align}
      In the middle equation, the first term vanishes thanks to the $\h$-part of the Bianchi identity, namely, $d\bOmh+[\varpih,\bOmh]+[\varpim,\bOmm]=0$, the second term is nothing but $\ad_{\varpih}P(\bOmh^n)$ which vanishes by $H$-invariance. The third term will be under discussion whether $P$ is $G$-invariant or if solely its $H$-invariance will be considered.
        
        According to the Chern-Simons theory, let $\varpi_t=t\varpi$ be an interpolating family of $\g$-valued connections, $0 \leq t \leq 1$, with curvature $\bOm_t= \bOm_{t,\h} \oplus \bOm_{t,\m} = d\varpi_t + \tfrac{1}{2} [\varpi_t,\varpi_t]$ together with the Bianchi identity $d\bOm_t + [\varpi_t ,\bOm_t]=0$. One has $\frac{d}{dt}\bOm_t = d\varpi + [\varpi_t,\varpi]$.
The standard Chern-Simons derivation leads, on the one hand, 
\begin{equation}\label{eq:CS-1}
    \frac{d}{dt} P(\bOm_t^n) = n P(\frac{d}{dt}\bOm_t, \bOm_t^{n-1}) = nP(d\varpi + [\varpi_t,\varpi],\bOm_t^{n-1})
\end{equation}
and on the other hand, 
\begin{equation}\label{eq:CS-2}
    n d P(\varpi, \bOm_t^{n-1}) = n P (d\varpi, \bOm_t^{n-1}) - n(n-1) P(\varpi, d\bOm_t,\bOm_t^{n-2}).
\end{equation}
The use of the Bianchi identity for $\bOm_t$ in the right hand side, relies on the $\Ad$-invariance of $P$, which in our situation at hand, depends on either the $G$ or the $H$-invariance of $P$.

\bigskip\noindent
Case 1. For $P\in{\cal I}^n(G)$, the $\Ad(G)$-invariance tells us that
\begin{equation}\label{eq:ad-inv}
    \ad_{\varpi_t} P(\varpi, \bOm_t^{n-1})= P([\varpi_t,\varpi],\bOm_t^{n-1}) - (n-1) P(\varpi,[\varpi_t,\bOm_t],\bOm_t^{n-2})=0
\end{equation}
and a straightforward computation yields a link with \eqref{eq:CS-1}
\begin{align*}
    n\,d P(\varpi, \bOm_t^{n-1}) &= nP (d\varpi, \bOm_t^{n-1}) - n(n-1) P(\varpi, d\bOm_t,\bOm_t^{n-2}) + \ad_{\varpi_t} P(\varpi, \bOm_t^{n-1}) \\
    &= nP(d\varpi + [\varpi_t,\varpi],\bOm_t^{n-1}) - n(n-1) P(\varpi, d\bOm_t+[\varpi_t,\bOm_t],\bOm_t^{n-2})\\
    &= n P(d\varpi + [\varpi_t,\varpi],\bOm_t^{n-1}) = \frac{d}{dt} P(\bOm_t^n) 
\end{align*}
where the Bianchi identity for $\bOm_t$ has been used. The transgression formula gives
\begin{equation} \label{eq:CS}
    P(\bOm^n) = \int_0^1 \frac{d}{dt} P(\bOm_t^n) dt = \int_0^1 n d P(\varpi, \bOm_t^{n-1}) dt = d Q_{2n-1}(\varpi).
\end{equation}
Restriction of $P\in{\cal I}^n(G)$ to the subgroup $H\subset G$, requires to consider the $\h$-part of the curvature from the decomposition $\bOm=\bOmh \oplus \bOmm$ to obtain the polynomial $P(\bOmh^n)$. Going back to \eqref{eq:dP}, since $P\in{\cal I}^n(G)$, one can use the $\ad (\m)$-invariance,
        \[
        \ad_{\varpim} P(\bOmm,\bOmh^{n-1}) = P([\varpim,\bOmm],\bOmh^{n-1}) + (n-1) P(\bOmm,[\varpim,\bOmh],\bOmh^{n-2}) =0.
        \]
    This latter relation allows to relax some conditions on torsion of the Cartan connection in order to secure \eqref{eq:dP} to be $0$ ($P(\bOmh^n)$ is closed) if one wishes to avoid to consider the $[\m,\m]$ contributions. Indeed, the Bianchi identity $d\bOmm + [\varpih,\bOmm] + [\varpim,\bOmh] = 0$ infers that $[\varpim,\bOmh]=0$ if the covariant derivative of the torsion $d\bOmm + [\varpih,\bOmm]
=0$. One thus gets
         \begin{align*}
           (i) &\ dP(\bOmh^n) = n(n-1) P([\varpim,\bOmh], \bOmh^{n-2},\bOmm) \\[2mm]
           (ii) &\ P(\bOmh^n) + \sum_{k=1}^n \begin{pmatrix} n \\ k \end{pmatrix} P(\bOmh^{n-k},\bOmm^k) = d Q_{2n-1}(\varpi) \ \mathrm{with}\ Q_{2n-1}(\varpi) = n \int_0^1 dt P(\varpi,\bOm^{n-1}_t)
       \end{align*}
this last equation comes from \eqref{eq:CS} by considering the expansion of $P(\bOm^n)=P((\bOmh\oplus\bOmm)^n)$.

       \bigskip\noindent
        Case 2. For  $P\in{\cal I}^n(H)$, only the $H$-invariance can be taken into account. Accordingly, one must restrict eqs\eqref{eq:CS-1} and \eqref{eq:CS-2} to the $\bOmh$ and $\bOm_{t,\h}$ components only and for which only the $\h$-part of the Bianchi identity, $(D\bOm)_\h =0$, can be used. Together with the result \eqref{eq:dP}, all in all, this gives
        \begin{align*}
            (i) & ~ dP(\bOmh^n) = - n P([\varpim,\bOmm],\bOmh^{n-1})\\
            (ii) & ~ \dfrac{d}{dt} P (\bOm_{t,\h}^n) = n d P(\varpih, \bOm_{t,\h}^{n-1})+n P([\varpi_{t,m},\varpim], \bOm_{t,\h}^{n-1}) + n(n-1) P (\varpih, [\bOm_{t,\m},\varpi_{t,\m}],\bOm_{t,\h}^{n-2})
        \end{align*}
        where the last two terms could be gathered as the result of $n\, \ad_{\varpi_{t,m}} P(\varpim,\bOmh^{n-1})$, see \eqref{eq:ad-inv}. But since $P\in{\cal I}^n(H)$, the invariance of $P$ under the $\m$-part of the symmetric Lie algebra $\g$ no longer holds true.
\end{proof}

As can be seen in $(a),(b)(c)$ and $(d)$, the vanishing of the $[\m,\m]$ commutators removes the obstruction to the Chern-Weil theorem in the case of symmetric Lie algebras. Furthermore, as stated in \cite[Proposition 4.1]{lotta2004model} there is an equivalence between a symmetric Lie algebra $\g=\h\oplus\m$ with $[\m,\m]\subset\h$ and an abelian reductive decomposition $\g'=\h\oplus\m'$ with $[\m',\m']=0$ through a mutation map and where $\m\simeq\m'$ as $\h$-modules, ({\em i.e.} $[\h,\m]\subset\m$ and $[\h,\m']\subset\m'$ and $\m\simeq\m'$ as vector spaces, see also \cite[Lemma 6.4, p.220]{sharpe2000differential}). Exploiting this fact, one has 
\begin{corollary}
\label{mm0 corrolary}
    Another way to reach conditions for the Chern-Weil theorem in both cases is to consider the unique (up to isomorphism)\footnote{See \cite[Lemma 6.4]{sharpe2000differential} and \cite[Proposition 4.1]{lotta2004model}.} mutation $\g \rightarrow \g'$ such that $[\m',\m'] = 0$.
\end{corollary}

In the Lorentzian case of Section \ref{section SO 5 ETC}, 
this mutation 
corresponds to setting either $k_1$ or $k_2=0$, thus obtaining the Minkowski geometry.

On the other hand, for the Lorentz$\times$Weyl geometry embedded into the mutated Möbius geometry which is a symmetric Lie algebra, the mutation map giving rise (up to isomorphism) to an abelian reductive decomposition could be
\begin{align}
		\label{mutation Möbius 2}
		\varphi_\g =
		\begin{pmatrix}
            z & \bar{a} & 0 \\
			b  & c	& a \\
			0 & \bar{b}  & -z 
		\end{pmatrix}
		\longrightarrow
		\varphi_{\g'} = \mu(\varphi_\g)=
        \left\{ 
		\begin{array}{ll}
		\begin{pmatrix}
			z & 0 & 0 \\
			\gamma_1 b  & c	&  \gamma_2 a \\
			0 & 0  & -z 
		\end{pmatrix} \in \g'' = \h\oplus \m'\\
        \\
        \begin{pmatrix}
			z & \bar{\gamma}_2 \bar{a} & 0 \\
			0  & c	&  0 \\
			0 & \bar{\gamma}_1 \bar{b}  & -z 
		\end{pmatrix} \in \g' = \h\oplus \m'' 
        \end{array}
        \right. 
\end{align}
where $[\m',\m']=[\m'',\m'']=0$ and $\m\simeq\m'\simeq\m''$ as $\h$-modules. They correspond respectively to \eqref{mutation Möbius} with $\bar{\gamma}_1=\bar{\gamma}_2=0$ or with $\gamma_1=\gamma_2=0$.

\smallskip
Therefore, in the Lorentzian case, see \eqref{eq:mm}, or the mutated Weyl geometry (as a subgeometry of the mutated Möbius geometry) through the identification $\alpha = \dfrac{k}{2} \beta$ together with \eqref{eq:mm-moebius}, 
we respectively have, on the one hand, 
\begin{align*}
\tfrac{1}{2} [\varpim,\varpim] = \left.
\begin{cases}
    \dfrac{k}{\ell^2} \begin{pmatrix}
			\beta\wedge \bar{\beta}& 0 \\
			0 & 0
		\end{pmatrix} =  
       - \dfrac{\Lambda_0}{3} \begin{pmatrix}
			\beta\wedge \bar{\beta}& 0 \\
			0 & 0
		\end{pmatrix}
       \\
       \\
        \dfrac{k}{\ell \ell'}  \begin{pmatrix}
			0 & 0 & 0 \\
			0 &  \beta \wedge \bar{\beta} & 0 \\
			0 & 0 & 0
		\end{pmatrix} =  
       - \dfrac{\Lambda_0}{3}
       \begin{pmatrix}
			0 & 0 & 0 \\
			0 &  \beta \wedge \bar{\beta} & 0 \\
			0 & 0 & 0
		\end{pmatrix}
\end{cases} \right\}= - \dfrac{\Lambda_0}{3} \dfrac{1}{2} \beta^a \wedge \beta^b J_{ab}
      \end{align*}
and, on the other hand,
\begin{align*}
    [\varpim,\bOmm] &= 
    \begin{cases}
        \dfrac{k}{\ell^2} \begin{pmatrix}
        \beta\wedge \bar{T} - T\wedge\bar{\beta} & 0 \\
        0 & 0 \end{pmatrix} = - \dfrac{\Lambda_0}{3} \begin{pmatrix}
        \beta\wedge \bar{T} - T\wedge\bar{\beta} & 0 \\
        0 & 0 \end{pmatrix} 
              \\
        \\
        \dfrac{k}{\ell \ell'}  \begin{pmatrix}
			0 & 0 & 0 \\
			0 &  \beta\wedge \bar{T} - T\wedge\bar{\beta} & 0 \\
			0 & 0 & 0
		\end{pmatrix} =  
       - \dfrac{\Lambda_0}{3}
            \begin{pmatrix}
			0 & 0 & 0 \\
			0 &  \beta\wedge \bar{T} - T\wedge\bar{\beta} & 0 \\
			0 & 0 & 0
		\end{pmatrix} \end{cases}  \\[2mm]
        &= - \dfrac{\Lambda_0}{3} \dfrac{1}{2} (\beta^a \wedge T^b - \beta^b \wedge T^a) J_{ab} . 
    \end{align*}
This shows that these $[\m,\m]$ contributions to $\h$ can be rendered very small in the limit $\Lambda_0 \to 0$\footnote{Which is a good approximation if one considers $\Lambda_\text{exp} \approx \Lambda_0$.} provided that the gauge fields and their derivatives have compact supports on $\mathcal{M}$.
This leads to a perturbative topological gauge theory in the parameter $\Lambda_0$ corresponding to a bare cosmological constant. This can also be seen as a physical implementation of the mutations. In this type of deformed topological gauge theories, the small observed value of the cosmological constant therefore finds a novel interpretation as a perturbatively topological action, with perturbative expansion in the cosmological constant.

\medskip
For example, taking the Lorentzian geometry of section \ref{section SO 5 ETC} we can build the total action:
    \begin{align}
		\label{topological action mutation 2}
		S_{top} &= v (S_G  + S_M) = v (S_G + \mathrm{Re} S_D) 
	\end{align}
    where we tuned 
    $v$ such that $v,v(\dfrac{2}{\gamma} - \dfrac{4 \Lambda_0 G r}{3\pi})$ tend towards $0$. In this way, we retrieve (if we assume $\mathcal{M}$ has no boundary or when neglecting boundary terms) the equations of motion of section \ref{section SO 5 ETC} with $\Lambda_0 \rightarrow 0$.
    Interestingly, we can make $\Lambda_0 \rightarrow 0^{\pm}$ such that we end up with a dS or AdS spacetime. To some extent, this limit seems to correspond to the interpretation as a "regulator" of the bare cosmological constant in the case of AdS spacetime \cite{Aros:1999id}.
    \begin{remark}
        One is still free to set $v$ such that the coupling constants of the Euler, Pontrjagin and Nieh-Yan terms in the action are respectively $-v\dfrac{3}{64\pi G\Lambda_0}$, $v (\dfrac{r}{16\pi^2} + \dfrac{3}{32\pi G \Lambda_0 \gamma})$, $\dfrac{v\Lambda_0r}{12 \pi^2}$ and have non-vanishing values. While the other terms that perturb topological invariance asymptotically vanish.
    \end{remark} 
 
	\section{Conclusion}

    We have shown that a large class of models for gravitation, including the standard theory of general relativity (modulo the topological Euler density associated to $R$ in the action), can be solely derived from invariant polynomials of the curvature. In particular, LQG with non-trivial Barbero-Immirzi parameter can be obtained by adding the Pontryagin class to the Pfaffian of the Lorentz-valued curvature. 
    In general, only the part related to the Pontryagin class $P(\bOm)$ is invariant under the choice of Cartan connection up to the integral over $\mathcal{M}$ of an exact term according to the Chern-Weil theorem.
    
    \medskip
    
    The following tables summarize the roles played by the different characteristic classes and invariant polynomials involved in the actions for the Lorentzian and Möbius geometries. These geometries were studied through mutations allowing to study symmetric types of Cartan geometries with group $H$. The parameters entering into the choice of mutation maps can be related to a bare cosmological constant $\Lambda_0$.

\medskip\noindent
    Mutated Lorentzian Cartan geometry $\bOmh=R + \dfrac{k}{\ell^2} \beta \wedge \bar{\beta}$ : $\Lambda_0 = - \dfrac{3 k}{\ell^2}$
    \begin{center}
			\begin{tabular}{ |c|c|c| } 
				\hline
				Characteristic class / & Terms & Theory \\
                Invariant polynomial &  & \\
				\hline
				$\Pf(\dfrac{\bOmh}{2\pi})$ & Palatini + $\Lambda_0$ + Euler of R & MM gravity \\ 
				\hline
				\multirow{2}{4em}{$\mathdutchcal{P}(\bOm)$} & \multirow{2}{13em}{Nieh-Yan + Pontryagin of R} & \multirow{2}{16em}{\centering Completion of MM gravity to LQG \\ \centering  with topological BI parameter} \\
				&  & \\ 
				\hline
                $\det (\mathds{1} + \dfrac{\bOmh}{2\pi}) $& $R^{ab}\beta_a \beta_b$ + Pontryagin of R & Completion of Palatini term to Holst \\
                 \hline
			\end{tabular}
	\end{center}

 \medskip\noindent
    Mutated Möbius Cartan geometry with $\alpha = \dfrac{k}{2} \beta$, $\bOmh=F+f$ : $\Lambda_0 = - \dfrac{3 k}{\ell \ell'}$
	\begin{center}
		\begin{tabular}{ |c|c|c| } 
			\hline
			Characteristic class / & Terms & Theory \\
            Invariant Polynomial & & \\
			\hline
			$\Pf(\dfrac{F}{2\pi})$  & Palatini + $\Lambda_0$ + Euler of R & MM gravity \\ 
			\hline
			\multirow{3}{4em}{$\mathdutchcal{P}(\bOm)$} & \multirow{3}{14em}{{\centering ~~Nieh-Yan + Pontrjagin of R} \\ {\centering 
            ~~~ ~~~ ~~~~~~+$ d\lambda \wedge d\lambda $}} & \multirow{3}{16em}{{\centering Completion of MM gravity to LQG} \\ {\centering ~~ with topological BI parameter} \\ {\centering ~
            ~~~ and dilational kinetic term
            }}  \\
			&  & \\ 
			& &   \\
			\hline
            $\det (\mathds{1} + \dfrac{\bOmh}{2\pi}) $& $R^{ab}\beta_a \beta_b$ + Pontryagin of R & Completion of Palatini term to Holst \\
            \hline
		\end{tabular}
	\end{center}

    \medskip
	Two classes of examples have been studied from which Einstein's equations (with positive or negative bare cosmological constant $\Lambda_0$) can be obtained.
	
	\smallskip
	The first example corresponds to the case of a $\g=			
	\left\{
	\begin{array}{ll}
		\mathfrak{so}(4,1) \text{ for } \Lambda_0 > 0 \\
		\mathfrak{so}(3,2) \text{ for } \Lambda_0 < 0
	\end{array}	
	\right.$-valued Cartan connection $\varpi$ defined on the  $H=SO(3,1)$-principal bundle $\mathcal{P}$. The computed action gives the Holst action together with the Euler, Pontrjagin, Nieh-Yan topological terms as well as a bare cosmological constant term, see \eqref{topological action mutation}. The coupling constants of these topological terms are entirely determined by the four parameters entering the Holst action ($G$, $\gamma$), the Nieh-Yan term ($r$) and the bare cosmological constant $\Lambda_0$.
	
	\smallskip

	The second example deals with the Möbius geometry. 
    In this case restricting the geometry to a $4$-dimensional submanifold $\mathcal{M}$ yields the same results as in the previous class of examples if one drops out dilations. Adding to the original action \eqref{eq:action-Moebius-bis} an interaction $T^a \wedge \lambda \wedge \beta_a$ between dilations and torsion instead leads to Einstein's equations (modified by the Holst term) with an additional source term (due to the dilational part of the algebra) for the curvature depending on a particular variation of spin density $-\dfrac{\pi G}{3} \partial_{\mu} \mathfrak{s}^{ab}_{~~ab\nu} \varepsilon^{\mu\nu}_{~~\;kc}$ and on torsion (if one does not assume $T=0$) when coupled to matter.
	
	\smallskip
	In these two classes of examples treated in the paper, the Pontrjagin number $\mathdutchcal{P}(\bOm)$ brings about a topological contribution to the Barbero-Immirzi parameter via the Nieh-Yan density, it also adds the other parity violating term (see \cite{date_topological_2009,Freidel:2005ak}) corresponding to the Pontrjagin density of the spacetime curvature $R$. 

    We also showed the construction of the action motivated by topological arguments leads to constrained couplings of the Euler/Gauss-Bonnet and Pontrjagin densities now expressed in terms of the previous four parameters corresponding to $G$ (gravitational coupling), $\gamma$ (BI parameter), $r$ (Nieh-Yan coupling) and $\Lambda_0$ (bare cosmological constant). The values of the constrained couplings are consistant with the specific values allowing for finite Noether charges as described in \cite{Aros:1999id,Miskovic:2009bm}.

    In the future we think studying the actions described in this manuscript in terms of BV equivalence \citep{cattaneo2024gravity} could also give very interesting results. The formulation of the deformed topological action \eqref{eq:action_start} in terms of deformations of Chern characteristic classes also motivates us to investigate it as a possible generalization of the Kodama state in gravitational theories.
    
    Finally, the last section addressed the issue of the topological character of the actions. It has been shown  that one can recover an action asymptotically invariant under the choice of gauge fields with compact supports and perturbative expansion in the bare cosmological constant $\Lambda_0$.

   In light of some recent developments in cosmology \cite{tang2024uniting,abdul2025desi} regarding the Hubble tension, the next step (explored in \cite{Thibaut:2025lwx}) shall be to consider spacetime dependent mutations, thus leading to a dynamical cosmological "constant".
 
	\subsubsection*{Acknowledgements}
	
	We thank our colleagues Thierry Masson, Thomas Schücker, Pietro Dona and Federico Piazza for helpful discussions.
    One of us (J.T) would also like to thank Thomas Ströbl for insightful comments.

\end{document}